\documentclass[a4paper,UKenglish]{article}
\pdfoutput=1
 
\usepackage{microtype}%if unwanted, comment out or use option "draft"
\usepackage[USenglish]{babel}
\usepackage{amsthm,amsmath,amssymb}
\usepackage{calc}
\usepackage{url}
\usepackage{booktabs}
\usepackage{enumitem}
\usepackage{fullpage}
\usepackage[sort&compress]{natbib}
\usepackage[utf8]{inputenc}
\usepackage{hyperref}

\usepackage[textwidth=20mm]{todonotes}

\newtheorem{theorem}{Theorem}
\newtheorem{corollary}{Corollary}
\newtheorem{lemma}{Lemma}
\newtheorem{definition}{Definition}
\newtheorem{hypothesis}{Hypothesis}
\newtheorem*{triangleconjecture}{Triangle Conjecture}
\newtheorem*{3sumconjecture}{3SUM Conjecture}
\newtheorem{proposition}{Proposition}

\newcommand{\eps}{\varepsilon}
\newcommand{\defparproblem}[4]{
	\bigskip
	\noindent
	\fbox{
		\begin{minipage}{.96\linewidth}
			\textsc{#1} \hfill Parameter: #2\\[2pt]
			\smallskip
			\noindent
			\begin{tabular}{@{}l@{ }l}
				\emph{Input:} & \begin{minipage}[t]{\linewidth-\widthof{Input:\ \ }}
					#3
				\end{minipage}\\[2pt]
				\emph{Task:} & \begin{minipage}[t]{\linewidth-\widthof{Input:\ \ }}
					#4
				\end{minipage}
			\end{tabular}
		\end{minipage}
	}
	\medskip
}

\newcommand\blfootnote[1]{%
  \begingroup
  \renewcommand\thefootnote{}\footnote{#1}%
  \addtocounter{footnote}{-1}%
  \endgroup
}

\title{Dynamic Parameterized Problems and Algorithms} 

\author{Josh Alman\thanks{MIT CSAIL, Cambridge, MA, USA. \texttt{jalman@mit.edu}. Supported by NSF Grant DGE-114747. Work initiated while at Stanford University.}
    \and
        Matthias Mnich\thanks{Universit{\"a}t Bonn, Institut f{\"u}r Informatik, Bonn, Germany \emph{and} Maastricht University, Department of Quantitative Economics, Maastricht, The Netherlands. \texttt{mmnich@uni-bonn.de}. Supported by ERC Starting Grant 306465 (BeyondWorstCase).}
    \and
        Virginia Vassilevska Williams \thanks{MIT CSAIL, Cambridge, MA, USA. \texttt{virgi@mit.edu}. Supported by NSF Grants CCF-141-7238, CCF-1528078 and CCF-1514339, and BSF Grant BSF:2012338. Work initiated while at Stanford University.}\blfootnote{A preliminary version of this paper appears in ICALP 2017.}
}

\begin{document}

\maketitle

\begin{abstract}
  Fixed-parameter algorithms and kernelization are two powerful methods to solve $\mathsf{NP}$-hard problems.
  Yet, so far those algorithms have been largely restricted to \emph{static} inputs.

  In this paper we provide fixed-parameter algorithms and kernelizations for fundamental $\mathsf{NP}$-hard problems with \emph{dynamic} inputs.
  We consider a variety of parameterized graph and hitting set problems which are known to have $f(k)n^{1+o(1)}$ time algorithms on inputs of size $n$, and we consider the question of whether there is a data structure that supports small updates (such as edge/vertex/set/element insertions and deletions) with an update time of $g(k)n^{o(1)}$; such an update time would be essentially optimal. 
  Update and query times independent of $n$ are particularly desirable.
  Among many other results, we show that {\sc Feedback Vertex Set} and\linebreak {\sc $k$-Path} admit dynamic algorithms with $f(k)\log^{O(1)}n$ update and query times for some function~$f$ depending on the solution size $k$ only.

  We complement our positive results by several conditional and unconditional lower bounds.
  For example, we show that unlike their undirected counterparts, {\sc Directed Feedback Vertex Set} and {\sc Directed $k$-Path} do not admit dynamic algorithms with $n^{o(1)}$ update and query times even for constant solution sizes $k\leq 3$, assuming popular hardness hypotheses.
  We also show that unconditionally, in the cell probe model, {\sc Directed Feedback Vertex Set} cannot be solved with update time that is purely a function of $k$.
\end{abstract}

\section{Introduction}
\label{sec:introduction}
The area of dynamic algorithms studies data structures that store a dynamically changing instance of a problem, can answer queries about the current instance and can perform small changes on it.
The major question in this area is, how fast can updates and queries be?

The most studied dynamic problems are dynamic graph problems such as connectivity (e.g.,~\cite{HenzingerKing2001,HolmEtAl2001,HuangEtAl2016,Thorup2000}), reachability~\cite{HenzingerEtAl2015b}, shortest paths (e.g.,~\cite{BernsteinRoditty2011,DemetrescuItaliano2004,HenzingerEtAl2016}), and maximum matching~\cite{BhattacharyaEtAl2015,GuptaPeng2013,Solomon2016}.
For a dynamic graph algorithm, the updates are usually edge or vertex insertions and deletions.
Any dynamic graph algorithm that can perform edge insertions can be used for a static algorithm by starting with an empty graph, and using~$m$ insertions to insert the $m$-edge input graph.
That is, if the update time of the dynamic algorithm is $u(m)$ then the static problem can be solved in $O(m\cdot u(m))$ time, plus the time to query for the output.
Hence, if a problem requires $\Omega(f(m))$ time to be solved statically, then any dynamic algorithm that can insert edges, and can be queried for the problem solution in $o(f(m))$ time, must need $\Omega(f(m)/m)$ (amortized) time to perform updates.
This is not limited to edge updates; similar statements are true for vertex insertions and other update types.
A fundamental question is which problems can be fully dynamized, i.e., have dynamic algorithms supporting updates in $O(f(m)/m)$ time where~$f(m)$ is the static runtime?

This question is particularly interesting for static problems that can be solved in near-linear time.
For them, we are interested in near-constant time updates---the holy grail of dynamic algorithms.
The field of dynamic algorithms has achieved such full dynamization for many problems.
A prime example of the successes of this vibrant research area is the {\em dynamic connectivity} problem: maintaining the connected components of a graph under edge updates, to answer queries about whether a pair of vertices is connected.
This problem can be solved with amortized expected update time $O(\log n \log\log^2 n)$~\cite{HuangEtAl2016,Thorup2000} and query time $O(\log n / \log\log\log n)$; polylogarithmic deterministic amortized bounds are also known, the current best by Wulff-Nielsen~\cite{Wulff-Nilsen13a}. 
After much intense research on the topic \cite{HenzingerK99,HenzingerT97,HolmEtAl2001}, the first polylogarithmic {\em worst case expected} update times were obtained by Kapron et al.~\cite{KapronEtAl2013}, who were the first to break through what seemed like an $\Omega(\sqrt n)$ barrier; the bounds of Kapron et al.~\cite{KapronEtAl2013} were recently improved by Gibb et al.~\cite{gibb2015dynamic}.
Similar $\tilde{O}(1)$ update and query time bounds\footnote{Throughout this paper, we write $\tilde{O}(f(n,k))$ to hide $\text{polylog}(n)$ factors.} are known for many problems solvable in linear time such as dynamic minimum spanning tree, biconnectivity and $2$-edge connectivity~\cite{HenzingerKing2001,HolmEtAl2001}, and maximal matching~\cite{BaswanaEtAl2015,Solomon2016}.

Barriers for dynamization have also been studied extensively.
Many unconditional, cell probe lower bounds are known.
For instance, for connectivity and related problems it is known~\cite{PatrascuDemaine2006,PatrascuThorup2011} that either the query time or the update time needs to be $\Omega(\log n)$.
However, current cell probe lower bound techniques seem to be limited to proving polylogarithmic lower bounds.
In contrast, conditional lower bounds based on popular hardness hypotheses have been successful at giving tight bounds for problems such as dynamic reachability, dynamic strongly connected components and many more~\cite{AbboudVassilevskaWilliams2014,HenzingerEtAl2015,KopelowitzEtAl2016,Patrascu2010}.

While the field of dynamic algorithms is very developed, practically all the problems which have been studied are polynomial-time solvable problems.
What about $\mathsf{NP}$-hard problems? Do they have fast dynamic algorithms?
By the discussion above, it seems clear that (unless $\mathsf{P}=\mathsf{NP}$), superpolynomial query/update times are necessary, and surely this is not as interesting as achieving near-constant time updates.
If the problem is relaxed, and instead of exact solutions, approximation algorithms are sufficient, then efficient dynamic algorithms have been obtained for some polynomial time approximable problems such as dynamic approximate vertex cover~\cite{BaswanaEtAl2015,BhattacharyaEtAl2015,OnakRubinfeld2010}.
What if we insist on exact solutions?

The efficient dynamization question does make sense for {\em parameterized} $\mathsf{NP}$-hard problems. For such problems, each instance is measured by its size $n$ as well as a \emph{parameter} $k$ that measures the optimal solution size, the treewidth or genus of the input graph, or any similar structural property.
If $\mathsf{P}\neq \mathsf{NP}$, then the runtime of any algorithm for such a problem needs to be superpolynomial, but it is desirable that the superpolynomiality is only in terms of~$k$.
That is, one searches for so-called {\em fixed-parameter algorithms} with runtime $f(k) \cdot n^c$ for some computable function $f$ and some fixed constant $c$ independent of $k$ and~$n$.
The {\em holy grail} here is an algorithm with runtime $f(k)\cdot n$ where~$f$ is a modestly growing function.
Such linear-time fixed-parameter algorithms can be very practical for small $k$.
The very active area of fixed-parameter algorithms has produced a plethora of such algorithms for many different parameterized problems.
Some examples include (1) many branching tree algorithms such as those for {\sc Vertex Cover} and {\sc $d$-Hitting Set}, (2) many algorithms based on color-coding~\cite{AlonEtAl1995} such as for {\sc $k$-Path}, (3) all algorithms that follow from Courcelle's theorem\footnote{Courcelle's theorem states that every problem definable in monadic second-order logic of graphs can be decided in linear time on graphs of bounded treewidth.}~\cite{Courcelle1990}, and (4) many more~\cite{BodlaenderEtAl2016,Dorn2010,EtscheidMnich2016,IwataEtAl2014,LokshtanovEtAl2015,vanBevern2014,Wahlstrom2014}.

We study whether $\mathsf{NP}$-hard problems with (near-)linear time fixed-parameter algorithms can be made efficiently dynamic. 
The main questions we address are:
\begin{itemize}
  \item Which problems solvable in $f(k) \cdot n^{1+o(1)}$ time have dynamic algorithms with update and query times at most $f(k) \cdot n^{o(1)}$?   
  \item Which problems solvable in $f(k) \cdot n$ time have dynamic algorithms with update and query times that depend solely on $k$ and not on $n$?
  \item Can one show that (under plausible conjectures) a problem requires $\Omega(f(k) \cdot n^\delta)$ (for constant $\delta>0$) update time to maintain dynamically even though statically it can be solved in $f(k) \cdot n^{1+o(1)}$ time?
\end{itemize}

\subsection{Prior Work}
We are aware of only a handful papers related to the question that we study.
Bodlaender~\cite{Bodlaender1993} showed how to maintain a tree decomposition of constant treewidth under edge and vertex insertions and deletions with $O(\log n)$ update time, as long as the underlying graph always has treewidth at most 2.
Dvo\v{r}\'{a}k et al.~\cite{DvorakEtAl2014} obtained a dynamic algorithm maintaining a tree-depth decomposition of a graph under the promise that the tree-depth never exceeds~$D$; edge and vertex insertions and deletions are supported in $f(D)$ time for some function~$f$.
Dvo\v{r}\'{a}k and T\r{u}ma~\cite{DvorakTuma2013} obtained a dynamic data structure that can count the number of induced copies of a given $h$-vertex graph, under edge insertions and deletions, and if the maintained graph has bounded expansion, the update time is bounded by $O(\log^{h^2} n)$.

A more recent paper by Iwata and Oka~\cite{IwataOka2014} gives several dynamic algorithms for the following problems, under the promise that the solution size never grows above $k$:
(1) an algorithm that maintains a {\sc Vertex Cover} in a graph under $O(k^2)$ time edge insertions and deletions and $f(k)$ time queries\footnote{Here $f(k)$ denotes the runtime of the fastest fixed-parameter algorithm for {\sc Vertex Cover} when run on $k$-vertex graphs.}, (2) an algorithm for {\sc Cluster Vertex Deletion} under $O(k^8+k^4\log n)$ time edge updates and $f(k)$ time queries\footnote{Here $f(k)$ denotes the runtime of the fastest fixed-parameter algorithm for {\sc Cluster Vertex Deletion} when run on $k^5$-vertex graphs.}, and (3) an algorithm for {\sc Feedback Vertex Set} in graphs with maximum degree $\Delta$ where edge insertions and deletions are supported in amortized time $2^{O(k)}\Delta^3\log n$.
Notably, when discussing {\sc Feedback Vertex Set}, the paper concludes: ``It seems an interesting open question whether it is possible to construct an efficient dynamic graph without the degree restriction.''

The final related papers are by Abu-Khzam et al.~\cite{Abu-KhzamEtAl2014,Abu-KhzamEtAl2015}.
Although these papers talk about parameterized problems and dynamic problems, the setting is very different.
Their problem is, given two instances $I_1$ and~$I_2$ of a problem that only differ in $k$ ``edits'', and a solution $S_1$ of $I_1$, to find a feasible solution $S_2$ of~$I_2$ that is at Hamming distance at most $d$ from $S_1$.
The question of study is whether such problems admit fixed-parameter algorithms for parameters $k$ and $\ell$.
That question though is not about data structures but about a single update.
Moreover, their algorithm is given $S_1$ as input, which---unlike a dynamic data structure---cannot force the initial solution to have any useful properties.
Thus, the hardness results in their setting do not translate to our data structure setting.
Furthermore, the runtimes in their setting, unlike ours, must have at least a linear dependence on the size of the input, as one has to at least read the entire input.

Besides the work on parameterized dynamic algorithms, there has been some work on parameterized streaming algorithms by Chitnis et al.~\cite{ChitnisCHM15}.
This work focused on {\sc Maximal Matching} and {\sc Vertex Cover}.
The difference between streaming and dynamic algorithms is that (a) the space usage of the algorithm is the most important aspect for streaming, and (b) in streaming, a solution is only required at the end of the stream, whereas a dynamic algorithm can be queried at any point and needs to be efficient throughout, but can use a lot of space.
For {\sc Vertex Cover} instances whose solution size never exceeds $k$, Chitnis et al.~\cite{ChitnisCHM15} give a one-pass randomized streaming algorithm that uses $O(k^2)$ space and answers the final query in $2^{O(k)}$ time; when the vertex cover size can exceed $k$ at any point, there is a one-pass randomized streaming algorithm using $O(\min\{m,nk\})$ space and answering the query in $O(\min\{m,nk\})+2^{O(k)}$ time. 

The relevant prior work on (static) fixed-parameter algorithms \cite{AlonEtAl1995,BeckerEtAl2000,BodlaenderEtAl2009,Buss,CaiEtAl2006,CaiEtAl1997,ChenEtAl2010,ChenEtAl2008,Cygan2012,Bedlewo2014,CyganEtAl2014,cygan2016known,DaligaultEtAl2010,DomEtAl2009,EstivillCastroEtAl2005,Fernau2006,GrammEtAl2009,Gyarfas1990,IwaideNagamochi2016,Iwata2017,KratschWahlstrom2012,LokshtanovEtAl2014,LokshtanovEtAl2016,vanBevern2014} is discussed in Section~\ref{sec:relatedwork}. Prior work on (not necessarily fixed-parameter) dynamic algorithms which we use in our algorithms is described in Section~\ref{sec:priorwork}.

\subsection{Our Contributions}
\label{sec:ourcontributions}

\textbf{Algorithmic Results.}
We first define the notion of a fixed parameter dynamic problem as a parameterized problem with parameter~$k$ that has a data structure supporting updates and queries to an instance of size $n$ in time $f(k)n^{o(1)}$.
The class $\mathsf{FPD}$ contains all such parameterized problems.
By a formalization of our earlier discussion, $\mathsf{FPD}$ is contained in the class of parameterized problems admitting algorithms running in time $f(k)n^{1+o(1)}$.
After this, we introduce two techniques for making fixed-parameter algorithms dynamic, and then use them to develop dynamic fixed-parameter algorithms for a multitude of fundamental optimization problems.
Our algorithmic contributions are stated in Theorem~\ref{algs} below.
In the runtimes, $DC(n)$ refers to the time per update to a dynamic connectivity data structure on~$n$ vertices, which from prior work (see Proposition \ref{thm:dynamicconnectivity} in Section \ref{sec:dynamicconnectivity}) can be:
	\begin{itemize}
	\item {\em expected amortized} update time $O(\log n(\log\log n)^2)$, or
	\item {\em expected worst case}\footnote{The expected worst case data structures for dynamic connectivity from the literature assume an \emph{oblivious adversary} who does not get access to the random bits used by the data structure, so our results using $DC$ with expected worst case guarantees do as well.} time $O(\log^4 n)$, or
	\item {\em deterministic amortized} time $O(\log^2 n/\log\log n)$.
	\end{itemize}
Which of these bounds we pick determines the type of guarantees (expected vs. deterministic, worst case vs. amortized) that the algorithm gives.
\begin{theorem}
\label{algs}
  The following problems admit dynamic fixed-parameter algorithms:
  \begin{itemize}
    \item {\sc Vertex Cover} parameterized by solution size under edge insertions and deletions, with $O(1)$ amortized or $O(k)$ worst case update time and $O(1.2738^k)$ query time,
    \item {\sc Connected Vertex Cover} parameterized by solution size  under edge insertions and deletions, with $O(k2^k)$ update time and $O(4^k)$ query time,
    \item {\sc $d$-Hitting Set} for all values of $d$ parameterized by solution size under set insertions and deletions, either with $O(kd^k)$ {\em expected} update time and $O(k)$ query time, or with $O(f(k,d))$ (worst-case, deterministic) update time and $O(d^{k} d! (k+1)^d)$ query time, for a function $f$ loosely bounded by $(d!)^d k^{O(d^2)}$.
  	\item {\sc Edge Dominating Set} parameterized by solution size under edge insertions and deletions, with $O(1)$ update time and $O(2.2351^k)$ query time,
    \item {\sc Feedback Vertex Set} parameterized by solution size under edge insertions and deletions, with $2^{O(k\log k)}\log^{O(1)} n$ amortized update time and $O(k)$ query time,
    \item {\sc Max Leaf Spanning Tree} parameterized by solution size under edge insertions and deletions, with $O(3.72^k + k^5\log n + DC(n))$ amortized update time, and maintains the current max leaf spanning tree explicitly in memory,
    \item {\sc Dense Subgraph in Graphs with Degree Bounded by $\Delta$} parameterized by the number of vertices in the subgraph under edge insertions and deletions, with $2^{O({k\Delta})} \cdot DC(n)$ update time and $2^{O(k\Delta)}\log n$ query time.
    \item {\sc Undirected $k$-Path} parameterized by the number of vertices on the path, with $k!2^{O(k)} \cdot DC(n)$ update time and $k!2^{O(k)}\log n$ query time.
    \item {\sc Edge Clique Cover} parameterized by the number of cliques and under the promise that the solution never grows bigger than $g(k)$, with $O(4^{g(k)})$ update time and $2^{2^{O(k)}} + O(2^{4g(k)})$ query time.
    \item {\sc Point Line Cover} and {\sc Line Point Cover} parameterized by the size of the solution and under point and line insertions and deletions, respectively, with $O(g(k)^3)$ update time and $O(g(k)^{2g(k) + 2})$ query time, under the promise that the solution never grows to more than $g(k)$.
  \end{itemize}
\end{theorem}

\medskip
\noindent
\textbf{Discussion of the Algorithmic Results.}
Our dynamic algorithm for {\sc Vertex Cover} and that of Iwata and Oka~\cite{IwataOka2014} both have query time $O(1.2738^k)$, by using the best known fixed parameter algorithm for {\sc Vertex Cover} on the maintained kernel.
However, our algorithm improves upon theirs in two ways.
First, our update time is amortized {\em constant} or $O(k)$ worst case, whereas the Iwata-Oka algorithm has update time $O(k^2)$.
Second, their update time bound of $O(k^2)$ only holds if the vertex cover is guaranteed to never grow larger than $k$ throughout the sequence of updates.
Namely, their update time depends on the size of their maintained kernel, which may become unbounded in terms of $k$.
Our algorithm does not need any such promise---it will always have fast (amortized $O(1)$ or $O(k)$ worst case) update time and return a vertex cover of size $k$ if it exists, or determine that one does not. 
This is a much stronger guarantee. 

Our dynamic algorithm and Chitnis et al.'s streaming algorithm for {\sc Vertex Cover} are both based on Buss' kernel, but our algorithm is markedly different from theirs.
In particular, we actually work with a modified kernel that allows us to achieve constant amortized update time. Because our algorithm is completely deterministic, it necessarily needs $\Omega(m)$ space, and our algorithm does indeed take linear space.

We give two algorithms for {\sc $d$-Hitting Set}.
The first is based on a randomized branching tree method, while the second is deterministic and maintains a small kernel for the problem.
For every constant~$d$, any {\sc $d$-Hitting Set} instance on $m$ sets and $n$ elements has a kernel constructible in time $O(dn + 2^d m)$ that has $O(d^{d+1} d! (k+1)^d)$ sets,
due to van Bevern~\cite{vanBevern2014}, and a kernel constructible in time $O(m)$ that has $O((k+1)^d)$ sets, due to Fafianie and Kratsch~\cite{FafianieKratsch2015}.
Unfortunately, it seems difficult to efficiently dynamize these kernel constructions.
Because of this, we present a {\em novel kernel} for the problem.
Our kernel can be constructed in $O(dn+3^d m)$ time and has size $O((d-1)! (k+1)^d)$.
It also has nice properties that make it possible to maintain it dynamically with update time that is a function of only~$k$ and $d$.
In fact, for any fixed $d$, the update time is polynomial in~$k$.

Our algorithm for {\sc Feedback Vertex Set} is a nice combination of kernelization and a branching tree.
Aside from our dynamic kernel for {\sc $d$-Hitting Set}, this is probably the most involved of our algorithms.
Iwata and Oka~\cite{IwataOka2014} had also presented a dynamic fixed-parameter algorithm for {\sc Feedback Vertex Set}.
However, their update time depends linearly on the maximum degree of the graph, and is hence efficient only for bounded degree graphs. Their paper asks whether one can remove this costly dependence on the degree. Our algorithm answers their question in the affirmative---it has fast updates regardless of the graph density.

All of our algorithms, except for the last two in the theorem, meet their update and query time guarantees regardless of whether the currently stored instance has a solution of size~$k$ or not.
The two exceptions, {\sc Edge Clique Cover} and {\sc Point Line Cover}, only work under the promise that the solution never grows bigger than a function of $k$.
In this sense they are similar to most of the algorithms from prior work~\cite{DvorakEtAl2014,IwataOka2014}.
There does seem to be an inherent difficulty to removing the promise requirement, however. In fact, in the parameterized complexity literature, these two problems are also exceptional, in the sense that their fastest fixed parameter algorithms run by computing a kernel and then running a brute force algorithm on it~\cite{cygan2016known,kratsch2016point}, rather than anything more clever.

\medskip
\noindent
\textbf{Hardness Results.}
In addition to the above algorithms, we also prove conditional lower bounds for several parameterized problems, showing that they are likely not in $\mathsf{FPD}$.
To our knowledge, ours are the first lower bounds for any dynamic parameterized problems. 

The hardness hypothesis we assume concerns Reachability Oracles (ROs) for DAGs: an RO is a data structure that stores a directed acyclic graph and for any queried pair of vertices $s,t$, can efficiently answer whether $s$ can reach $t$.
(An RO does not perform updates.) 
Our main hypothesis is as follows:
\begin{hypothesis}[RO Hypothesis]
  On a word-RAM with $O(\log m)$ bit words, any Reachability Oracle for directed acyclic graphs on $m$ edges must either use $m^{1+\eps}$ preprocessing time for some $\eps>0$, or must use $\Omega(m^\delta)$ time to answer reachability queries for some constant $\delta>0$.
\end{hypothesis}

The only known ROs either work by computing the transitive closure of the DAG during preprocessing, thus spending $\Theta(\min\{mn,n^\omega\})$ time (where $n$ is the number of vertices and $2\leq \omega<2.373$ \cite{legallmult,VassilevskaWilliams2012}), or by running a BFS/DFS procedure after each query, thus spending $O(m)$ time.
Both of these runtimes are much larger than our assumed hardness; hence the RO Hypothesis is very conservative. 

We also use a slightly weaker version of the RO Hypothesis, asserting that its statement holds true even restricted to DAGs that consist of $\ell$ layers of vertices (for some fixed constant~$\ell$), so that the edges go only between adjacent layers in a fixed direction, from layer $i$ to layer~$i+1$.
While this new {\em LRO Hypothesis} is certainly weaker, we show that it is implied by either of two popular hardness hypotheses: the {\sc 3SUM} Conjecture and the Triangle Conjecture.
The former asserts that when given $n$ integers within $\{-n^c,\ldots,n^c\}$ for some constant $c$, deciding whether three of them sum to $0$ requires $n^{2-o(1)}$ time on a word-RAM with $O(\log n)$ bit words.
The latter asserts that detecting a triangle in an $m$-edge graph requires $\Omega(m^{1+\eps})$ time for some $\eps>0$.
These two conjectures have been used for many conditional lower bounds~ \cite{AbboudVassilevskaWilliams2014,GajentaanOvermars2012,KopelowitzEtAl2016}.

P{\v{a}}tra{\c{s}}cu studied the RO Hypothesis, and while he was not able to prove it, the following strong cell probe lower bound follows from his work~\cite{Patrascu2011}:  there are directed acyclic graphs on $m$ edges for which any RO that uses $m^{1+o(1)}$ preprocessing time (and hence space) in the word-RAM with $O(\log n)$ bit words, must have $\omega(1)$ query time.
Using this statement, unconditional, albeit weaker lower bounds can be proven as well.
This is what we prove:
\begin{theorem}
  Fix the word-RAM model of computation with $w$-bit words for $w = O(\log n)$ for inputs of size~$n$.
  Assuming the LRO Hypothesis, there is some $\delta>0$ for which the following dynamic parameterized graph problems on $m$-edge graphs require either $\Omega(m^{1+\delta})$ preprocessing or $\Omega(m^\delta)$ update or query time:
  \begin{itemize}
   \item {\sc Directed $k$-Path} under edge insertions and deletions,
   \item {\sc Steiner Tree} under terminal activation and deactivation, and
   \item {\sc Vertex Cover Above LP} under edge insertions and deletions. 
  \end{itemize}
  Under the RO Hypothesis (and hence also under the LRO Hypothesis), there is a $\delta>0$ so that {\sc Directed Feedback Vertex Set} under edge insertions and deletions requires $\Omega(m^\delta)$ update time or query time. 

  Unconditionally, there is no computable function $f$ for which a dynamic data structure for {\sc Directed Feedback Vertex Set} performs updates and answers queries in $O(f(k))$ time. 
\end{theorem}

Our lower bounds show that, although {\sc $k$-Path} and {\sc Feedback Vertex Set} have fixed parameter dynamic algorithms for undirected graphs, they probably do not for directed graphs.
Interestingly, the fixed-parameter algorithms for {\sc $k$-Path} in the static setting work similarly on both undirected and directed graphs, so there only seems to be a gap in the dynamic setting.

All problems for which we prove lower bounds have $f(k)n^{1+o(1)}$ time static algorithms, except for {\sc Vertex Cover above LP}.
However, it seems that the reason why the current algorithms are slower is largely due to the fact that near-linear time algorithms for maximum matching are not known.
Recent impressive progress on the matching problem~\cite{Madry13} gives hope that an $f(k)n^{1+o(1)}$ time algorithm for {\sc Vertex Cover above LP} might be possible.

A common feature of most of the problems above is that they are either not known to have a polynomial kernel (like {\sc Directed Feedback Vertex Set}), or do not have one unless $\mathsf{NP} \subseteq \mathsf{coNP}/poly$ (like {\sc $k$-Path}~\cite{BodlaenderEtAl2009} and {\sc Steiner Tree} parameterized by the number of terminal pairs~\cite{DomEtAl2009}).
One might therefore conjecture that problems which cannot be made fixed parameter dynamic do not have polynomial kernels, or vice versa.
Tempting as it is, this intuition turns out to be false.
{\sc Vertex Cover Above LP} does not have a dynamic fixed-parameter algorithm, yet it is known to admit a polynomial kernel~\cite{KratschWahlstrom2012}.
On the other hand, the {\sc $k$-Path} problem on undirected graphs also does not admit a polynomial kernel unless $\mathsf{NP} \subseteq \mathsf{coNP}/poly$~\cite{BodlaenderEtAl2009}, yet we give a dynamic fixed-parameter algorithm for it.
Hence, the existence of a polynomial kernel for a parameterized problem is not related to the existence of a dynamic fixed-parameter algorithm for it.

\medskip
\noindent
\textbf{Preliminaries.}
We assume familiarity with basic combinatorial algorithms, especially graph algorithms and hitting set algorithms.
When referring to a graph $G$, we will write $V(G)$ to denote its vertex set and~$E(G)$ to denote its edge set.
Unless otherwise specified, $n$ and $m$ will refer to the number of vertices and edges in $G$, respectively.
By $\tilde{O}(f(n))$ we denote $f(n) \log^{O(1)} n$.
We also assume familiarity with dynamic problems and parameterized problems.
For formal definitions, see Section~\ref{sec:formaldefs}.

\section{Overview of the Algorithmic Techniques}
\noindent
\textbf{Promise model and Full model.}
There are two different models of dynamic parameterized problems in which we design algorithms: the \emph{promise model} and the \emph{full model}.
When solving a problem with parameter~$k$ in the promise model, there is a computable function $g : \mathbb{N} \to \mathbb{N}$ such that one is promised that throughout the sequence of updates, there always exists a solution with parameter at most $g(k)$.
Hence, one only needs to maintain a solution under updates with good guarantees on both query and update times as long as the promise continues to hold.
If at any point during the execution no solution to the parameterized problem with parameter $g(k)$ exists, then the algorithm is not required to provide any guarantees.

In the full model, there is no such promise.
One needs to efficiently maintain a solution with parameter at most $k$, or the fact that no such solution exists, under any sequence of updates.
When possible, it is desirable to have an algorithm with guarantees in the full model instead of only the promise model, and all but two of our algorithms ({\sc Point Line Cover} and {\sc Edge Clique Cover}) do work in the full model.

\subsection{Techniques for designing dynamic fixed-parameter algorithms}

We present two main techniques for obtaining dynamic fixed-parameter algorithms: dynamic kernels and dynamic branching trees.

\medskip
\noindent
\textbf{Dynamization via kernelization.}
Using the notation of Cygan et al.~\cite{CyganEtAl2015}, a {\em kernelization algorithm} for a parameterized problem~$\Pi$ is an algorithm $\mathcal A$ that, given an instance $(I,k)$ of~$\Pi$, runs in polynomial time and returns an instance $(I', k')$ of $\Pi$ such that the size of the new instance is bounded by a computable function of $k$ and so that $(I', k')$ is a `yes' instance of $\Pi$ if and only if $(I,k)$ is. Frequently, when the problem asks us to output more than just a Boolean answer, then an answer for $(I',k')$ must be valid for $(I,k)$ as well.
We will refer to the output of $\mathcal A$ as a {\em kernel}.
For example, a kernelization algorithm for {\sc Vertex Cover} might take as input a graph $G$, and return a subgraph $G'$ such that any vertex cover of~$G'$ is also a vertex cover of~$G$.

In the first approach, we compute a kernel for the problem, and maintain that this is a valid kernel as we receive updates. In other words, as we receive updates, we will maintain what the output of a kernelization algorithm $\mathcal A$ would be, without actually rerunning $\mathcal A$ each time.
Similar to kernelizations for static fixed-parameter algorithms, if we can prove that the size of our kernel is only a function of $k$ whenever a solution with parameter $k$ exists, then we can answer queries in time independent of $n$ by running the fastest known static algorithms on the kernel.

The difficult part, then, is to efficiently dynamically maintain the kernel.
The details of how efficiently we can handle updates to the kernel also determines which model of dynamic fixed-parameter algorithm our algorithm works for.
If the kernel is defined by sufficiently simple or local rules such that updates can take place in time independent of the current kernel size, then the algorithm should work in the full model.
If updates might take time linear in the kernel size, then the algorithm only works in the promise model.

As we will see, there are many problems for which we can efficiently maintain a kernel.
In some instances we will be able to maintain the classical kernels known for the corresponding static problem, while in others, we will design new kernels which are easier to maintain.

\medskip
\noindent
\textbf{Dynamization via branching tree.}
In the second approach, we consider so-called {\em set selection problems}.
In these problems, the instance consists of a set of objects $U$ (e.g., vertices of a graph), the parameter is $k$, and one needs to select a subset $S\subseteq U$ of size $k$ (at least~$k$/at most $k$) so that a certain predicate~$P(S)$ is satisfied.
Many parameterized problems are of this nature, such as {\sc $k$-Path}, {\sc Vertex Cover}, and {\sc (Directed) Feedback Vertex Set}.

Consider a (static) set selection problem which admits a branching solution.
By this we mean, for every instance $U$ of the problem, there is an `easy to find' subset $T\subseteq U$ of size $|T|\leq f(k)$ (for some function~$f$) so that any solution $S$ of size at most $k$ must intersect $T$.
Furthermore, for any choice of $t\in T$ to be placed in the solution, one can efficiently obtain a reduced instance of the problem with parameter $k-1$, which corresponds to picking $t$ to be in the solution.
For instance, for {\sc Vertex Cover}, every edge $\{u,v\}$ can be viewed as such a subset $T$ since at least one of $u$ and $v$ is in any vertex cover, and if we pick~$u$, then we can remove it and all its incident edges from the graph to get a reduced instance.

For such problems, there is a simple fixed-parameter algorithm called the branching tree algorithm: The algorithm can be represented as a tree $\mathcal T$ rooted at a node $r$ (we refer to the vertices of $\mathcal T$ as nodes).
Each node $v$ of the tree corresponds to a reduced instance of the original one, and in this instance, $v$ has a subset $T$ of size~$f(k)$, and a child $v_i$ for every $i\in T$, where $v_i$ corresponds to selecting $i$ to be placed in the solution, and $v_i$ carries the reduced instance where $i$ is selected.
The height of the tree $X$ is bounded by $k$ since at most $k$ elements need to be selected, and the branching factor is $f(k)$. 
Each leaf $\ell$ of the tree $\mathcal T$ is either a ``yes''-leaf (when the predicate is satisfied on the set of elements selected on the path from $r$ to $\ell$) or a ``no''-leaf (when the predicate is not satisfied).
The runtime of the algorithm bounded by $f(k)^k\cdot t(N)$, where~$t(N)$ is the time to find a subset $T$ that must contain an element of the solution in instances of size $N$, together with the time to find a reduced instance, once an element is selected.

What we have described so far is a static algorithmic technique, but we investigate when this algorithmic technique can be made dynamic.
In other words, given an update, we would like to quickly update $\mathcal T$ so that it becomes a valid branching tree for the updated instance.
Since the number of nodes in the branching tree is only a function of $k$, one can afford to look at every tree node. 
Ideally, one would like the time spent per node to only depend on $k$.
However, for most problems that we consider, the branching tree needs to be rebuilt every so often, since the subset $T$ to branch on may become invalid after an update, and the time to rebuild can have a dependence on the instance size.
We use two methods to avoid this.
The first is to randomize the decisions made in the branching tree (e.g., which set $T$ to pick) so that, assuming an oblivious adversary that must provide the update sequence in advance, it is relatively unlikely that we need to rebuild the tree $\mathcal T$ (or its subtrees) after each update, and in particular, so that the expected cost of an update is only a function of $k$.
The second is to make `robust' choices of~$T$, so that many updates are requires before the choice of $T$ becomes incalid, and then amortize the cost of rebuilding the tree over all the updates required to force such a rebuilding.

\subsection{Algorithm Examples} \label{subsec:examples}
We give overviews of the techniques used in some of our algorithms, to demonstrate the dynamic kernel and dynamic branching tree approaches, and different ways in which they can be used.
We emphasize that these descriptions are substantial simplifications which hide many non-trivial details and ideas.

\medskip
\noindent
\textbf{Vertex Cover.}
We give both a dynamic kernel algorithm and a dynamic branching tree algorithm for {\sc Vertex Cover}.

Our first algorithm maintains a kernel obtained as follows:
Every vertex of degree at least $k+1$ `selects' $k+1$ incident edges arbitrarily and adds them to the edge set $E'$ of the kernel, independently of other vertices.
Next, every edge incident to two vertices of degree at most $k$ is also added to $E'$.
Finally, the vertices set of the kernel consists of all vertices that are not isolated in~$E'$. 
This is a valid kernel, since any vertex cover of size at most $k$ needs to include every vertex of degree strictly greater than $k$. Every edge in $E'$ either has both its end points of degree at most $k$, or is selected by one of its end points of degree at least $k+1$. 
Any node of a vertex cover of the kernel either has degree at most $k$ or selects $k+1$ edges.
Thus
the kernel must have size $O(k^2)$ when a vertex cover of size at most $k$ exists.
To insert an edge we simply add the edge to the kernel unless one of its incident vertices has degree greater than $k$.
If one of the end points $x$ used to be of degree at most $k$ and is now of degree $k+1$, we have $x$ select all its incident edges and add them to the kernel.
To delete an edge,
%To remove an edge, 
 we simply remove it from the kernel.
 If it was incident to a vertex $v$ of degree higher than $k+2$, then we need to find another edge incident to $v$ which is in the graph but not selected by $v$ to put into the kernel.
 If one of the end points now has degree $k$, we need to go through the incident edges and remove them from the kernel if their other end point has high degree and did not select them.
All these operations can be performed by storing appropriate pointers so that the updates run in $O(k)$ time
 With a little bit more work one can make them run in $O(1)$ amortized time.
To answer queries, we answer ``no'' in constant time if the kernel has more than $2k(k+1)$ edges, and otherwise 
we run the fastest fixed-parameter algorithm for static {\sc Vertex Cover} on the kernel of size $O(k^2)$ and parameter $k$.
This results in a $O(1.2738^k)$ update time.

Our second algorithm maintains a branching tree of depth at most $k$, which corresponds to using following randomized branching strategy: pick a uniformly random edge, and branch on adding each of its endpoints into the vertex cover.
For a static branching algorithm, there is no need to pick a uniformly random edge to branch on, since at least one endpoint of every single edge must be in the vertex cover.
However, a deterministic branching strategy like this in a dynamic algorithm would be susceptible to an adversarial edge update sequence, in which the adversary frequently removes edges which have been chosen to branch on.
By ensuring that each edge we branch on is a uniformly random edge, we make the probability that we need to recompute any subtree of the branching tree~$\mathcal T$ low.
We compute the expected update time to be only $O(k 2^k)$.
Queries can be answered in only $O(k)$ time by following a path in the branching tree to an accepting leaf, if one exists.

These algorithms demonstrate some subtleties of the two techniques.
In the branching tree algorithm, we use a randomized branching rule to deal with adversarial updates.
In some of our other branching tree algorithms, like for {\sc Feedback Vertex Set}, we are able to find a deterministic branching rule to yield a deterministic algorithm instead.
In the kernelization algorithm, we manage to find a kernel which can be updated quickly even when the answer becomes larger than $k$ and the kernel size becomes large.
In other problems, it will be harder to do this, and we may need to restrict ourselves to the promise model where we are guaranteed that the kernel will not grow too big in order to have efficient update times.
Dynamic kernelization techniques typically lead to faster update times and query times, like in this case, because we can apply the fastest known static algorithm for the problem to the kernel to answer queries.
In a branching tree algorithm, we may be using a branching rule which does not lead to the fastest algorithm because it is easier to dynamically maintain.

Interestingly, we are able to generalize both of these algorithms to the {\sc $d$-Hitting Set} problem.
The {\sc $d$-Hitting Set} branching tree algorithm is similar to that of {\sc Vertex Cover}, but the {\sc $d$-Hitting Set} dynamic kernelization algorithm is much more complicated, and involves a tricky recursive rule for determining which sets to put in the kernel.

\medskip
\noindent
\textbf{Max Leaf Spanning Tree.}
Our algorithm for {\sc Max Leaf Spanning Tree} uses the dynamic kernel approach.
The kernel we maintain is simply the given graph, where we contract vertices of degree two whose neighbors both also have degree two.
We can maintain this kernel by storing paths of contracted vertices in lists corresponding to edges they have been contracted into.
As long as this kernel has $\Omega(k^2)$ vertices, it must always have a spanning tree with at least $k$ leaves.

Unlike in other dynamic kernel algorithms, where we maintain that the kernel does not get too large, this kernel may grow to have $\Omega(n^2)$ edges.
We can nonetheless find a tree with at least~$k$ leaves in time independent of $n$, by breadth-first searching from an arbitrary vertex in the kernel until we have $\Omega(k^2)$ kernel vertices, and just finding a tree within those vertices.

This method finds a subtree $T_S$ with at least $k$ leaves, but we need to find a tree which spans the whole graph.
In the static problem, this could be accomplished by a linear-time breadth-first search away from~$T_S$, but in the dynamic problem, this is too slow.
To overcome this, we also maintain a spanning tree~$T$ of the entire graph, which does not necessarily have~$k$ leaves, using a known dynamic tree data structure.
When queried for a spanning tree, we find $T_S$, and then perform a `merge' operation to combine~$T$ and~$T_S$ into a spanning tree with at least $k$ leaves.
This merge operation makes only $O(k^4)$ changes to~$T$, so we are able to maintain a desired spanning tree in time independent of $n$.

We are able to maintain a linear size answer in only logarithmic time per update because the output is not very `sensitive' to updates: we can always output an answer very close to~$T$, which itself only changes in one edge per update.
In other problems where the output can be more sensitive to updates, like {\sc Edge Clique Cover}, we need to maintain a small intermediate representation of the answer instead of the answer itself.

\medskip
\noindent
\textbf{Feedback Vertex Set in undirected graphs.}
Our algorithm for {\sc Feedback Vertex Set} combines the dynamic kernel approach with the dynamic branching tree approach.
We will maintain a branching tree, where we branch off of which vertex to include in our feedback vertex set.
Then, at each node in the branching tree, we will maintain a kernel to help decide what vertices to branch on.
Similar to the situation with {\sc Max Leaf Spanning Tree}, our kernel can possibly have~$\Omega(n^2)$ edges.
Here we will deal with this by branching off of only $O(k)$ vertices in the kernel to add to our feedback vertex set, so that we can answer queries in sublinear time in the kernel size.

The kernel we maintain at each node of the branching tree is the given graph, in which vertices of degree one are deleted, and vertices of degree two are contracted.
This involves many details for maintaining contracted trees, and dealing with resulting self-loops.
Since the resulting graph has average degree at least three, whereas forests have much lower average degree, we show that a feedback vertex set of size at most $k$ must contain a vertex of high degree, whose degree is at least $1/(3k)$ of the total number of edges in the kernel.
Since there are at most~$6k$ such vertices, we can branch on which to include in our feedback vertex set.

This branching strategy works well for the static problem, but it is hard to maintain dynamically. 
Each edge update might change the set of vertices with high enough degree to branch on, and changing which vertex we branch on, and recomputing an entire subtree of the branching tree, can be expensive.
We alleviate this issue using amortization.
Instead of branching only on the $6k$ highest degree vertices, we instead branch on the $12k$ highest degree vertices.
If our kernel has~$m$ edges, then we prove that $\Omega(m)$ edge updates need to happen before there might be a small feedback vertex set containing none of the vertices we branched on.
After these updates we need to recompute the branching tree, but this is inexpensive when amortized over the required $\Omega(m)$ updates.

\section{Related work} 
\label{sec:relatedwork}
Here we discuss the best known fixed-parameter algorithms and kernelizations for the problems we study.

For {\sc Vertex Cover} parameterized by solution size $k$, we use an algorithm with runtime $O(1.2738^k + kn)$ by Chen et al.~\cite{ChenEtAl2010}, as well as a linear-time algorithm producing a kernel of size $O(k^2)$ due to Buss~\cite{Buss}.

For {\sc Vertex Cover Above LP} parameterized by the difference $\lambda$ between an optimal integral solution and an optimal solution to the natural LP relaxation of the problem, an algorithm with runtime $2.3146^\lambda\cdot n^{O(1)}$ was given by Lokshtanov et al.~\cite{LokshtanovEtAl2014}, and a randomized kernel of size $O(\lambda^3)$ is known~\cite{KratschWahlstrom2012}.
The variant of {\sc Vertex Cover} in which the subgraph induced by the vertex cover has to be connected can be solved in time $2^kk \cdot O(n+m)$~\cite{Cygan2012}, where $k$ again is the solution size.
This problem does not admit a kernel of size polynomial in $k$, unless $\mathsf{NP}\subseteq \mathsf{coNP}/poly$~\cite{DomEtAl2009}.

The generalization of {\sc Vertex Cover} from graphs to hypergraphs with hyperedges of size $\leq d$ is known as {\sc $d$-Hitting Set}.
A simple search tree algorithm solves this problem in time $d^k\cdot n^{O(1)}$, where $k$ is the desired size of the hitting set.
A linear-time algorithm~\cite{vanBevern2014} produces a kernel with $O(k^d)$ hyperedges; with additional $O(k^{1.5d})$ time the number of vertices can be reduced to $O(k^{d-1})$.

For the {\sc Feedback Vertex Set} problem parameterized by solution size $k$, a randomized $O(4^k \cdot kn)$ time algorithm due to Becker et al.~\cite{BeckerEtAl2000} finds a feedback vertex set of size $k$ with constant probability independent of $k$.
A linear-time algorithm by Iwata~\cite{Iwata2017} leads to a kernel of size $O(k^2)$.

For {\sc $k$-Path}, the color-coding method Alon et al.~\cite{AlonEtAl1995} leads to a randomized fixed-parameter algorithm with expected runtime $2^{O(k)}\cdot n$.
The {\sc $k$-Path} problem does not admit a polynomial kernel, unless $\mathsf{NP} \subseteq \mathsf{coNP}/poly$~\cite{BodlaenderEtAl2009}.

For {\sc Steiner Tree}, the relation between the number $n$ of vertices in the input graph and the terminal set size $k$ makes different algorithms the fastest choice: for $k\leq 4\log n$ an algorithm with runtime $3^k\cdot O(n) + 2^k\cdot O(m + n\log n))$~\cite{EricksonEtAl1987}, for $4\log n < k < 2\log n(\log\log n)^3$ an algorithm with runtime $2^{k + (k/2)^{1/3}(\ln n)^{2/3}}$~\cite{FuchsEtAl2007}, for $2\log n(\log\log n)^3 < k < (n-\log^2n)/2$ an algorithm with runtime $2^{k+\log_2k\log_2n}O(kn)$~\cite{Vygen2011}, and for $k > (n-\log^2n)/2$ brute force enumeration of minimum spanning trees for $T\cup X$ for each $X\subseteq V(G)\setminus T$ in time $2^{n-k}\cdot O(m\alpha(m,n))$, where $\alpha(m,n)$ is the inverse of the Ackermann function. 
Also this problem does not admit a kernel of size polynomial in $k$, unless $\mathsf{NP} \subseteq \mathsf{coNP}/poly$~\cite{BodlaenderEtAl2009}.

The complexity of {\sc Directed Feedback Vertex Set} was a long-standing open problem, until its fixed-parameter tractability was shown by Chen et al.~\cite{ChenEtAl2008}; an algorithm with runtime $4^kk!k^5\cdot O(n+m)$ is due to Lokshtanov et al.~\cite{LokshtanovEtAl2016}.
Whether this problem admits a kernel of size polynomial in $k$ is still unresolved~\cite{Bedlewo2014}.

Regarding {\sc Edge Clique Cover} parameterized by solution size $k$, a kernel with $2^k$ vertices obtainable in linear time is due to Gy\'{a}rf\'{a}s~\cite{Gyarfas1990}.
The instance can then be solved by brute force on the kernel, in time $2^{2^{O(k)}} + O(n^4)$, as described by Gramm et al.~\cite{GrammEtAl2009}; such doubly-exponential runtime is best possible assuming the Exponential-Time Hypothesis~\cite{cygan2016known}.
The result that this problem does not admit a kernel of size $k^{O(1)}$ is due to Cygan et al.~\cite{CyganEtAl2014}.

For {\sc Dense Subgraph} in graphs with maximum degree $\Delta$, a randomized fixed-parameter algorithm with expected runtime $2^{O(\Delta + k)}\cdot O((\Delta+k)n)$ was suggested by Cai et al.~\cite{CaiEtAl2006}.

For {\sc Point Line Cover}, it is folklore to derive a kernel with $O(k^2)$ points in linear time.

For {\sc Edge Dominating Set} parameterized by solution size $k$, a kernel with $O(k^2)$ vertices is known~\cite{Fernau2006}; the fastest fixed-parameter algorithm runs in $2.2351^k\cdot n^{O(1)}$ time~\cite{IwaideNagamochi2016}.

For {\sc Max Leaf Spanning Tree}, the fastest known algorithm in terms of solution size $k$ runs in $3.72^k \cdot n^{O(1)}$ time~\cite{DaligaultEtAl2010}.
An algorithm by Cai et al.~\cite{CaiEtAl1997} solves the problem in time $O((n + m) + (4(k+2)(k+1))^k)$.
The smallest known kernel has at most $3.75k$ vertices~\cite{EstivillCastroEtAl2005}.

\section{Dynamic Parameterized Problems - Formalizations}
\label{sec:formaldefs}

\medskip
\noindent
\textbf{Dynamic Problems.}
We will now formally define the types of problems we study.
A static computational problem~$\Pi$ defines a function $f_{\Pi}$ defined on $n$ bit inputs, and an algorithm for $\Pi$ computes $f_{\Pi}$ on any $n$-bit input, in time $t(n)$.
A {\em dynamic problem} $\Pi$ additionally has an update rule that describes small allowable changes to an $n$-bit instance $x$.
We say that $\Pi$ is \emph{$c(n)$-dynamic} if the update rule only changes at most $c(n)$ bits of the instance; typically, $c(n)=O(\log n)$.
If it is possible to get from any valid instance to any other via this update rule, the problem is called {\em fully dynamic}.
The {\em diameter} of a fully dynamic problem $\Pi$, denoted $D(\Pi)(n)$, is the maximum number of updates needed to get from any instance of size $n$ to any other instance of size $n$.
All of the problems we consider will be fully dynamic.

For example, the dynamic graph connectivity problem is a fully dynamic problem, where instances are graphs on $n$ vertices.
An instance on $m$ edges can be described using $2m\log (n+1)$ bits,\footnote{The $i$th edge is represented via the concatenation of the descriptions of its vertices; $(0^{\log n},0^{\log n})$ can be chosen as an empty edge.} and each update inserts or deletes a single edge (changing $2\log n$ bits in the description at a time). 
This problem is fully dynamic with diameter $2m$
since one can get from any $m$-edge graph to any other by entirely deleting all $m$ edges of the first graph and then inserting all $m$ edges of the second.

A {\em dynamic algorithm} for a dynamic problem $\Pi$ stores an $n$-bit instance $x$ of $\Pi$ and can perform updates to $x$ following the update rules of the problem. 
Dynamic algorithms store~$x$ in a data structure, which is allowed to store more than just $x$, and must support the updates in some (amortized or worst-case) time~$u(n)$, and be able to answer queries about~$f_\Pi(x)$ at any point in time.

If $\Pi$ is a fully dynamic problem, then the algorithm is called \emph{fully dynamic}.
Any fully dynamic algorithm for $\Pi$ with update time $u(n)$ also implies the existence of an $O(D(\Pi)(n)\cdot u(n))$ static algorithm for $\Pi$: start from a trivial instance $x_0$ and then update~$x_0$ at most $D(\Pi)(n)$ times in the dynamic algorithm to obtain $x$. For instance, for any graph problem, one can start from an empty graph and insert all edges of any desired graph $G$, one at a time, and hence a dynamic algorithm for a problem on $m$-edge graphs with update time~$u(m)$ implies a static algorithm running in time $O(m\cdot u(m))$. 

\medskip
\noindent
\textbf{Parameterized Problems.}
A \emph{parameterized problem} $\Pi$ is a subset of $\{0,1\}^\star \times \mathbb N$.
Thus, we consider bivariate inputs $(I,k)$ consisting of an instance $I$ together with a \emph{parameter} $k\in\mathbb N$.
We follow the general assumption that the parameter is part of the input, and $k = |I|^{O(1)}$.

We say that problem $\Pi$ is \emph{fixed-parameter tractable} if it admits a \emph{fixed-parameter algorithm}, which is an algorithm $\mathcal A$ that decides membership in $\Pi$ in time $f(k)\cdot |I|^{O(1)}$ for some computable function~$f$.

A \emph{almost fixed-parameter dynamic algorithm} for a parameterized dynamic problem $\Pi$ is a dynamic algorithm for $\Pi$ that performs updates in time $f(k)\cdot |I|^{o(1)}$.
A \emph{fixed-parameter dynamic algorithm} for a parameterized dynamic problem $\Pi$ is a dynamic algorithm for $\Pi$ that performs updates in time $f(k)$.
$\Pi$ is  \emph{(almost) fixed-parameter dynamic} if it admits a (almost) fixed-parameter dynamic algorithm.

Note in particular that when considering a problem in the fully dynamic model parameterized by~$k$, we assume that the parameter $k$ does not change in the dynamic process.
That said, a generic fully dynamic parameterized algorithm can be modified in a straightforward way to maintain, for instance, an answer for all parameters up to a given value $k$.

We define two related complexity classes: the class $\mathsf{FPD}$ of fixed-parameter dynamic parameterized problems, and the class almost-$\mathsf{FPD}$ of almost fixed-parameter dynamic parameterized problems.

We will focus on two subclasses of these: Lin-$\mathsf{FPD}$ is the subclass of $\mathsf{FPD}$ consisting of dynamic problems~$\Pi$ with $D(\Pi)(n)\leq O(n)$. Similarly, Lin-almost-$\mathsf{FPD}$ is the subclass of almost-$\mathsf{FPD}$ consisting of dynamic problems $\Pi$ with $D(\Pi)(n)\leq O(n)$. 

Notice that Lin-$\mathsf{FPD}$ is a subclass of those parameterized problems with $f(k)n$ time algorithms, and Lin-almost-$\mathsf{FPD}$ is a subclass of those with $f(k)n^{1+o(1)}$ runtimes.

We study which parameterized problems with $f(k)n$ or $f(k)n^{1+o(1)}$ time algorithms with a dynamic version with linear diameter lie in Lin-$\mathsf{FPD}$ and Lin-almost-$\mathsf{FPD}$, respectively (and thus in~$\mathsf{FPD}$).

\section{Hitting Set-like problems}

\subsection{Vertex Cover}

\defparproblem{Vertex Cover}{$k$}{An undirected graph $G$ and an integer $k\in\mathbb N$.}{Find a set $S$ of at most $k$ vertices intersecting all edges of $G$.}

{\sc Vertex Cover} is a canonical fixed-parameter tractable problem.
There are several ways to obtain a small kernel for {\sc Vertex Cover}.
For instance, a kernel with at most $2k$ vertices can be obtained via the standard LP relaxation.
Our construction will make use of a very simple construction of a kernel with $O(k^2)$ vertices, due to S. Buss~(cited by J. Buss and Goldsmith~\cite{BussGoldsmith1993}). 

The dynamic problem we will solve is maintaining a vertex cover of size at most $k$ (if one exists) under {\em edge insertions and edge deletions}.
Each query asks to return a vertex cover of size at most~$k$ of the current graph, or to report that none exists.
Our algorithm will work in the full model, so that our update time will be efficient regardless of whether there is a vertex cover of size $k$ or not, and we will be able to distinguish these two cases.

\subsubsection{Dynamic kernel for Vertex Cover}

Iwata and Oka~\cite{IwataOka2014} present a dynamic algorithm for {\sc Vertex Cover} using a dynamic kernel of size $O(k)$ that can be updated in $O(k^2)$ time in the promise model.
The kernel they maintain is the output of a greedy $2$-approximation algorithm for {\sc Vertex Cover}.
Their update time depends on the size of the kernel; in the promise model this is $O(k^2)$, but it can be $\Omega(n^2)$ in the full model.
Hence their dynamic algorithm does not have fast update time when there is no promise that the size of a minimum vertex cover is bounded above by a function of $k$.

We will use Sam Buss' kernel for {\sc Vertex Cover}. We will modify it slightly and will make it dynamic in the full model. We achieve query time $O(1.2738^k)$ and update time either amortized $O(1)$ (independent of $k$) or worst case $O(k)$, in both cases improving upon the result by Iwata and Oka. 

Buss' static kernelization algorithm places all vertices of degree at least $k+1$ into the desired vertex cover~$C$ of $G$ with size at most $k$, and removes their incident edges and any isolated vertices.
If $G$ has a vertex cover of size at most $k$, then the remaining graph, called~$K$, has $O(k^2)$ vertices and edges, since any vertex cover of $K$ with size $k$ covers all edges, but all vertices in $K$ have degree at most $k$. 

We define an equivalent kernel on $O(k^2)$ vertices and edges that is easier to make dynamic.
Let $E'$ be the edge set and $V'$ be the vertex set of the kernel $K$ that we will build.
We will define $E'$ and will then let $V'$ be the vertices in $V$ that have nonzero degree in $E'$.

Consider every vertex $v\in V$ that has degree at least $k+1$ in $G$.
Each such $v$ selects a set of $k+1$ incident edges $S(v)$ arbitrarily.
Each $v$ picks its edges independently from other vertices; in particular, $v$ might pick $\{u,v\}$, but $u$ might not.
We add $\cup_v S(v)$ to $E'$.
In addition, we add to $E'$ all edges incident to two vertices of degree at most $k$.
Notice that a low degree vertex might not have all its incident edges in $E'$.
%The vertices $V'$ of $K$ will be the non-isolated vertices in the subgraph of $G$ induced by $E'$.

Note that $K$ is a kernel by the same argument as in Buss' kernel.
If an edge is not in $E'$, then this is because it was not picked by a vertex $x$ of degree at least $k+1$, but $x$ must be selected in any vertex cover of size at most $k$ solution in $G$ and in $K$, so any vertex cover of~$K$ with size at most $k$ is a vertex cover of $G$ (with size at most $k$) as well.
Also, similar to Buss' kernel, $K$ has at most $O(k^2)$ vertices and edges whenever it admits a vertex cover $S$ of size at most $k$, which can be seen as follows.
Let $S$ be a vertex cover of $K$ (and hence~$G$) with size at most $k$.
Then all edges are incident to $S$.
We count the number of edges in~$E'$.
Consider any $v$ of degree at least $k+1$; its incident edges in $E'$ are of two types---those in $S(v)$, of which there are $k+1$, and those edges $\{v,x\}$ whose other endpoint $x$ is also of degree at least $k+1$ and is hence also in $K$, and such that $x$ has selected $\{v,x\}$.
Thus the number of edges in $E'$ incident to vertices of degree at least $k+1$ is at most $k(k+1)$.
The remaining edges are incident to low degree vertices in $S$, and there can be at most $k^2$ of them.
Thus, $|E'|\leq k(k+1)$. Clearly then also $|V'|\leq O(k^2)$ since $V'$ does not contain isolated vertices in $E'$.

In order to dynamically maintain this kernel, we will maintain, for every vertex $v$:
\begin{itemize}
  \item its degree $d_G(v)$ in $G$,
  \item the edge set $E'$ of the kernel $K$,
  \item if $v$ has degree at least $k+1$, a doubly linked list $R_v$ of (pointers to) the edges incident to $v$ that are not in $S(v)$\footnote{In our amortized update algorithm, some vertices of degree at least $k+1$ might not have a list $R_v$.},
  \item if $v$ has degree at least $k+1$, a pointer $r(q,v)$ for every $v$ and $q\in N(v)$ into the position of $\{v,q\}$ in $R_v$ (or \texttt{nil} if $\{q,v\}\in S(v)$),
  \item the vertex set $V'$ of $K$ that consists of the vertices $v\in V(G)$ for which $d_{E'}(v)\geq 1$.
\end{itemize}

\subsubsection{\texorpdfstring{$O(k)$}{O(k)} worst case update time}
Let us describe the update procedures for $K$ when we want to achieve $O(k)$ worst case update time.
In these procedures, we say that a vertex has \emph{high degree} if its degree is at least $k+1$; and \emph{low degree} otherwise.

Suppose that an edge $\{u,v\}$ is inserted into $G$.
First we increment $d_G(u)$ and $d_G(v)$.
There are a few cases.
First, if $d_G(u)\leq k$ and $d_G(v)\leq k$, then we insert $\{u,v\}$ into~$E'$.
Otherwise, suppose that $d_G(u)>k$.
(We process $v$ similarly if it has $d_G(v)>k$.)
If $d_G(u)>k+1$, then~$u$ was already a high degree vertex.
We just place $\{u,v\}$ into $R_u$, updating the pointer~$r(v,u)$.
If $d_G(u)=k+1$, then $u$ used to be a low degree vertex and now is a high degree vertex.
Then~$u$ must pick all of its incident edges and place them in~$S(x)$; we do this by adding every incident edge of $u$ to $E'$ and creating an empty $R_u$.
If $u$ or $v$ or any of their neighbors had no incident edges in $E'$ but now they do, we insert them into~s$V'$.
The total insert time is worst case~$O(k)$.

Suppose now that an edge $\{u,v\}$ is to be deleted from $G$.
First we decrement $d_G(u)$ and~$d_G(v)$ and delete $\{u,v\}$ from $E'$ if necessary.
Let us consider $u$ ($v$ will be processed similarly).
If $d_G(u)\geq k+1$ and $\{u,v\}\in R_u$, then we just remove $\{u,v\}$ from $R_u$ using the pointer~$r(v,u)$.
If $d_G(u)\geq k+1$ and $\{u,v\}\notin R_u$, then take the first edge $\{u,x\}$ in $R_u$, remove it from $R_u$ and add it to $E'$.
If $d_G(u)=k$, then $u$ used to be a high degree vertex and now is a low degree vertex.
We note that before the deletion of $\{u,v\}$, all of $u$'s incident edges must have been in~$S(u)$ and thus now all its incident edges are in $E'$.
However, if an incident edge $\{u,y\}$ has its other endpoint $y$ be a high degree vertex and $y$ has not selected $\{u,y\}$, then we need to remove $\{u,y\}$ from $E'$.
We do exactly that---we go through all of the~$k$ incident edges of $u$ and remove any from $E'$ that should not be there.
Finally, if $d_G(u)<k$, then $u$ was already a low degree vertex, so we do not need to fix anything.
The total deletion time is worst case~$O(k)$.

\subsubsection{\texorpdfstring{$O(1)$}{O(1)} amortized update time}
Now we explain how to achieve $O(1)$ amortized update time.
We will modify the kernel we are maintaining slightly.
Now we will have three types of vertices:
\begin{itemize}
  \item \emph{high degree}, those of degree at least $2k+1$ in $G$,
  \item \emph{low degree}, those of degree at most $k$ in $G$, and
  \item \emph{medium degree}, all others.
\end{itemize}
The high degree vertices will behave similar to previously---they will select between $k+1$ and $2k+1$ edges to add to $E'$.
The low degree vertices will also behave similar to before---they will only add edges to $E'$ that have their other endpoint of low degree or medium degree (or if their other end point selected the edge). 

The medium degree vertices will behave either like low degree or like high degree vertices, depending on the operations performed on them.
We say that a medium degree vertex $u$ \emph{looks like a low degree vertex} if it does not have a list $R_u$.
Otherwise, if it has a list $R_u$ (potentially \texttt{nil}), we say that it \emph{looks like a high degree vertex}.

To insert an edge $\{u,v\}$, we update the degrees of $u$ and $v$ and add $\{u,v\}$ to $E'$ if 
neither endpoint is of high degree or of medium degree that looks like high degree.
%both endpoints have low or medium degree.
Otherwise, we process any endpoint that either has high degree or looks like high degree.

Suppose $u$ is a high degree vertex or looks like a high degree vertex.
The first option is that $u$ has degree at least $k+1$ and has $R_u$ set (if its degree is strictly larger than $2k+1$,~$R_u$ will always be set as before the insertion $u$ had high degree).
We check how many edges are in $S(u)$ (we can keep a count, or calculate it from the degree and the size of $R_u$).
If their number is at most $2k$, we just add $\{u,v\}$ to $E'$; otherwise, we add it to $R_u$.

Suppose finally that $u$ has degree exactly $2k+1$ and $R_u$ does not exist.
This means that~$u$ used to be a medium degree vertex that looked like a low degree vertex, and is now a high degree vertex.
In this case we need to create $R_u$.
We go through all edges $\{u,x\}$ incident to~$u$, put them in $S(u)$ by setting~$R_u$ to \texttt{nil} and add them to $E'$.
Here we can afford to spend $O(k)$ time since at least~$k$ insertions must have occurred since $u$ was last a low degree vertex.

To delete an edge $\{u,v\}$, we decrement the degrees, remove $\{u,v\}$ from $E'$ and process both $u$ and~$v$.
Consider $u$.
If $u$ is a high degree vertex, then we proceed as before---if $\{u,v\}$ was not in $S(u)$, then we just remove it from $R_u$, and otherwise, we replace it in $E'$ with the first edge in $R_u$.
If $u$ was a high degree vertex but is now a medium degree vertex, we treat it as if it was a high degree vertex; now $u$ is a medium degree vertex that looks like a high degree vertex.
If $u$ is a medium degree vertex and also was a medium degree vertex, then we treat it like a high degree vertex if it looks like a high degree vertex, and otherwise we do nothing.
If $u$ is a low degree vertex and also used to be a low degree vertex, we do nothing.
If $u$ was a medium degree vertex and now is a low degree vertex, there are two options.
If~$u$ looked like a low degree vertex, we do nothing.
Otherwise, if $u$ looked like a high degree vertex, we know that since now its degree is $k$, all its incident edges are in $S(u)$ (and hence~$E'$) and $R_u$ is \texttt{nil}.
We delete~$R_u$ and we go through all of $u$'s incident edges $\{u,x\}$.
We remove each such $\{u,x\}$ from $E'$ if $x$ is a high degree vertex or a medium degree vertex that looks like a high degree vertex and $x$ did not select $\{u,x\}$ in $S(x)$.
Here we can afford an $O(k)$ runtime for the following reason: $u$ looked like a high degree vertex right before the deletion so we know that it was a high degree vertex at some point and between the last time that it became a high degree vertex and now at least $k$ deletions must have occurred.
Overall we obtain $O(1)$ amortized update time.

To answer a query, we first check whether $K$ has at most $2k(k+1)$ edges and at most $2k(k+2)$ vertices and if not, we return ``no''.
Otherwise we solve the {\sc Vertex Cover} problem on $K$ using the fastest known fixed-parameter algorithm for {\sc Vertex Cover} by Chen et al.~\cite{ChenEtAl2010}, that on $K$ runs in $O(1.2738^k)$ time.

\subsubsection{Dynamic algorithm via a dynamic branching tree}
Given $G$, we will keep a binary branching tree~$\mathcal T$, where each non-leaf node of $\mathcal T$ contains an edge to branch on.
Every non-leaf subtree $\mathcal T'$ of $\mathcal T$ represents a subgraph~$H$ of $G$ and the root of each $\mathcal T'$ contains an edge of its corresponding subgraph $H$.
If a tree node $t$ contains edge~$\{x,y\}$ and the subtree $\mathcal{T}_t$ under it corresponds to $H$, then the left subtree of $t$ corresponds to putting $x$ in the vertex cover and contains the edges of $H$ that are not covered by $x$, and the right subtree corresponds to putting $y$ in the vertex cover and contains the edges of $H$ not covered by $y$.

We will maintain the invariant that every non-leaf node $v$ of $\mathcal T$ contains an edge chosen uniformly at \emph{random} among the edges in the subgraph corresponding to the subtree under $v$.

A node of $\mathcal T$ becomes a leaf if its subgraph is empty, in which case it is a ``yes''-node, or if it is at depth~$k$, so that $k$ vertices have been added to the vertex cover along the path from the root to it, and then it is a ``no''-node if the subgraph contains any edges and a ``yes''-node otherwise.

Suppose that a new edge $\{x,y\}$ is inserted.
This is how to update $\mathcal T$:
Traverse each vertex of $\mathcal T$ starting at the root as follows:
suppose that we are at vertex $v$ such that the subgraph~$H$ corresponding to its subtree contains $L$ edges.
Then, after adding $\{x,y\}$ the number of edges is increased to $L+1$. 
We need to maintain the invariant that the edges in the nodes of~$\mathcal T$ are random from their subtrees.
Hence, with probability $1/(L+1)$ we replace the edge in~$v$ with $\{x,y\}$ and rebuild {\em the entire subtree} of $\mathcal T$ under~$v$.
If $\{x,y\}$ is not added into $v$, then the edge $e$ in $v$ has probability $(1-1/(L+1))\cdot (1/L) = 1/(L+1)$ of being into~$v$, and so the invariant is maintained.

How long does this take?
If $v$ is at distance $i$ from the root, then the number of vertices under it is at most~$2^{k-i}$ since there are only $k-i$ branching decisions left.
At each such node~$v$ at distance $i$ from the root, a random edge is placed and for each of the two children, at most $L$ edges need to be looked at to check whether they are covered.
So it takes $O(L 2^{k-i})$ time to rebuild the subtree.
However, this only happens with probability $1/(L+1)$, so in expectation the runtime is $O(2^{k-i})$.
There are $2^i$ nodes at distance $i$ from the root, so level~$i$ of the tree is rebuilt in expected time $O(2^k)$ and the entire tree in expected time $O(k2^k)$.

This rebuilding ensures that each edge in a subtree appears with the same probability at the root of the subtree maintaining our invariant.

To find a vertex cover of size $k$, one only needs to traverse the $2^k$ leaves of the tree and find a ``yes''-node if one exists (if it does not, just return ``no'' as there is no vertex cover of size at most $k$).
After finding the ``yes''-node, just go up the tree to the root to find the vertices chosen to be in the vertex cover.
In fact, we can improve the query time by maintaining in addition a pointer to a ``yes''-leaf node, if one exists.
Then the query time is $O(k)$.

Now suppose that edge $\{x,y\}$ is deleted.
Suppose first that $\{x,y\}$ does not appear anywhere in~$\mathcal T$.
Then consider any subtree of $\mathcal T$ under some tree node $v$ such that $\{x,y\}$ is in the subgraph $H$ corresponding to the subtree.
The edge $\{x',y'\}$ stored in $v$ is chosen uniformly at random from the edges of $H$. 
When we remove $\{x,y\}$, then $\{x',y'\}$ is still uniformly random among the remaining edges of $H$, so we do not have to do anything.

Now assume that $\{x,y\}$ does appear in $\mathcal T$.
At each node $v$ of $T$ containing $\{x,y\}$, we need to remove $\{x,y\}$ from $v$ and rebuild the entire subtree under $v$.
We will do this top to bottom starting at the root, searching for tree nodes containing $\{x,y\}$ in, say, a breadth first search fashion.

Consider one such $v$ containing $\{x,y\}$.
Suppose that the subgraph $H$ corresponding to it contains~$L$ edges and that $v$ is at distance $i$ from the root.
Similar to before, we remove~$\{x,y\}$ and rebuild the subtree in time $O(L2^{k-i})$.
However, $\{x,y\}$ only had a $1/L$ chance of being in node $v$ to begin with, so that we only spend $O(2^{k-i})$ time rebuilding that node in expectation.
Summing over all vertices we still get an expected time of $O(k2^k)$.

In summary, we have a dynamic algorithm for {\sc Vertex Cover} with expected worst case update time $O(k2^k)$ and $O(k)$ query time. 
This algorithm has a much higher update time and is randomized and thus the guarantee is only valid if the adversary is required to supply the updates offline.
Nevertheless, this algorithm has an essentially optimal query time.

\subsection{Connected Vertex Cover}
\defparproblem{Connected Vertex Cover}{$k$}{A graph $G$ and an integer $k\in\mathbb N$.}{Find a vertex cover $X\subseteq V(G)$ of $G$ with size $|X|\leq k$ inducing a connected subgraph of~$G$.}

We first describe a kernel for the problem and then show how to dynamize it.
The particular kernel we build on is suggested in Exercise 2.14 in the book by Cygan et al.~\cite{CyganEtAl2015}.
Let $S$ be the set of nodes in~$G$ of degree at least $k+1$; clearly, $S$ must be contained in any vertex cover of $G$ with size at most ~$k$.
Let~$Q$ be the set of nodes in $V(G)\setminus S$ that have no neighbors in $V(G)\setminus S$ (but have neighbors in~$S$).
As with the kernel for {\sc Vertex Cover}, the number of edges within $V(G)\setminus S$ and the number of vertices in $V(G)\setminus (S\cup Q)$ is~$O(k^2)$.
The vertices of $Q$ only have neighbors in $S$, and while they can be in a connected vertex cover, they are only there to make sure that the subgraph induced by the vertex cover is connected.
Thus, if two nodes in $Q$ have exactly the same neighbors in $S$, then we only need one of them in the kernel.
The kernel thus is the graph induced by (1) the nodes in $S$, (2) for every $s\in S$ a set of $k+1$ nodes adjacent to $s$, (3) the nodes in $V(G)\setminus (S\cup Q)$, and (4) for every $Y\subseteq S$ a vertex $q\in Q$ such that $N(q)=Y$ (if such a vertex exists).
The size of the kernel is $O(2^k)$.

To make this kernel dynamic, we store:
\begin{enumerate}
  \item  for every node $v$ its degree $d_G(v)$ in $G$ and a doubly-linked list $N(v)$ of its neighbors (where in addition, every $x$ that is a neighbor of $v$ has a pointer to its position in $N(v)$); 
  \item $S=\{v~|~v\in V(G), d_G(v)>k\}$ as a doubly-linked list (where in addition we have for each $v$, a pointer $s(v)$ to its position in $S$, or \texttt{nil} if $v\notin S$);
  \item for every $v\notin S$, the degree $d_{G-S}(v)$ of $v$ in the subgraph induced by $V(G)\setminus S$;
  \item a doubly-linked list $L$ of the nodes with $d_{G - S}(v)>0$ (where in addition we keep a pointer for each $v$ from $v$ to its position in $L$, or \texttt{nil} if $v$ is not in $L$);
  \item for every $Y\subseteq S$, a list $L_Y$ of nodes $q\notin S$ with $N(q)=Y$ (here again we keep for every~$Y$ and every $q$ a pointer from $q$ to its position in $L_Y$, or \texttt{nil} if $q\notin L_Y$).
\end{enumerate}

To answer a query, form the graph $K$ induced by the $O(2^k)$ nodes of $S$, $L$, for every $s\in S$ the first $k+1$ nodes in~$N(s)$, and the first node in $L_Y$ for each $Y\subseteq S$; then solve the {\sc Connected Vertex Cover} problem on $K$ for parameter $k$, say in $O(4^k)$ time by branching.

Suppose that we need to insert an edge $\{x,y\}$.
We first increment $d_G(x)$ and $d_G(y)$, and add~$y$ to $N(x)$ and $x$ to~$N(y)$; this takes $O(1)$ time.
If $d_G(x)$ became at least $k$, add $x$ to~$S$; similarly with $y$. 
Suppose that $x$ was added to $S$ (do the same for $y$ if $y$ was added to~$S$).
First, for every neighbor $a$ of~$x$ (there are exactly $k+1$ of these since $x$ was just added to~$S$), if $a\in V(G)\setminus S$ (which we can check by seeing if $s(a)=\texttt{nil}$), decrement $d_{G-S}(a)$ and if $d_{G-S}(a)=0$, remove $a$ from $L$.
This takes $O(k)$ time.
Then, go through all subsets $Y\subseteq S$ where $x\in Y$ and for every neighbor $a$ of~$x$, check whether $d_{G-S}(a)=0$ and $\{a,y\}\in E(G)$ for all $y\in Y$ but $\{a,z\}\notin E(G)$ for all $z\in S\setminus Y$, and if so, add $a$ to $L_Y$ ($L_Y$ was initially empty).
Note that there are at most $2^k-1$ such sets $Y$ and $k + 1$ neighbors $a$ of $x$.
This can be done in total $O(k2^k)$ time by following pointers.
Finally, if $x,y$ are in $V(G)\setminus S$, increment $d_{G-S}(x)$ and $d_{G-S}(y)$.
If~$x$ had $d_{G-S}(x)=0$ and now $d_{G-S}(x)=1$, we add $x$ to $L$, and
we remove $x$ from $L_Y$, where~$Y$ is the set of neighbors of $x$ excluding $y$.
Handle $y$ similarly.
The total insertion time is $O(k2^k)$.

Deletions are symmetric. 
Suppose that we need to delete an edge $\{x,y\}$.
We first decrement~$d_G(x)$ and~$d_G(y)$, and remove $y$ from $N(x)$ and $x$ from $N(y)$; this takes $O(1)$ time.
If $d_G(x)$ decreased to a value of at most $k$, remove $x$ from $S$; similarly with $y$. 
Suppose that $x$ was removed from $S$ (do the same for $y$ if $y$ was removed from $S$.).
First, for every neighbor $a$ of $x$ (there are exactly $k$ of these since $x$ was just removed from $S$), if $a\in V(G)\setminus S$ (which we can check by seeing if $s(a)=\texttt{nil}$), increment $d_{G\setminus S}(a)$ and if $d_{G\setminus S}(a)=1$, add $a$ to $L$.
Also, by going through the neighbors of $x$ we compute $d_{G-S}(x)$.
This all takes $O(k)$ time.
Then, go through all subsets $Y\subseteq S\cup \{x\}$ where $x\in Y$, and delete $L_Y$; there are at most $2^k - 1$ such sets $Y$.
Finally, if $x,y$ are in $V(G)\setminus S$ and neither was removed from $S$, decrement $d_{G-S}(x)$ and $d_{G - S}(y)$.
If $x$ had $d_{G - S}(x)=1$ and now $d_{G-S}(x)=0$, we remove $x$ from $L$, and add $x$ to~$L_Y$ where $Y$ is the set of neighbors of $x$.
Handle $y$ similarly.
The total deletion time is $O(k2^k)$.

\subsection{Edge Dominating Set}
For a graph $G$, an \emph{edge dominating set} is a set $D\subseteq E(G)$ of edges such that each edge of $G$ either belongs to $D$ or is incident to some edge in $D$.

\defparproblem{Edge Dominating Set}{$k$}{An undirected graph $G$ and an integer $k\in\mathbb N$.}{Find an edge dominating set of $G$ with size at most $k$.}

To show a dynamic algorithm, we use the well-known (see, e.g., Fernau~\cite{Fernau2006}) relation between edge dominating sets and vertex covers, and then apply our dynamic fixed-parameter algorithm for {\sc Vertex Cover} based on a dynamic kernel. The relation is that to every {\sc Edge Dominating Set} instance $G$ with a solution $D\subseteq V(G)$ of size at most $k$, there is a corresponding vertex cover of size at most $2k$, which consists of the endpoints of the edges in~$D$.

\begin{theorem}
  There is a dynamic algorithm for {\sc Edge Dominating Set} which handles updates in $O(1)$ time and queries in $O(2.2351^k)$ time.
\end{theorem}
\begin{proof}
  For a graph $G$, we will maintain a dynamic kernel $K$ together with the set $S\subseteq V(G)$ of vertices of degree greater than~$2k$.
  We can maintain $K$ and $S$ in just $O(1)$ time per edge update.
  The kernel $K$ that we are dynamically maintaining has the properties that
  \begin{itemize}
    \item if $G$ has a vertex cover with size at most $2k$ then $K$ contains at most $4k^2 + 2k$ vertices, and
    \item every vertex cover of $G$ of size at most $2k$ is equal to the union of $S$ and a subset of $K$.
  \end{itemize}
  Since an edge dominating set of $G$ with size at most $k$ is among vertices of a vertex cover of~$G$ with size at most $2k$, this means that
  \begin{itemize}
    \item if there is an edge dominating set of size at most $k$ then $K$ contains at most $4k^2 + 2k$ vertices and $S$ contains at most $2k$ vertices, and
    \item every edge dominating set of size at most $k$ of our graph is composed of edges among vertices in $K \cup S$.
  \end{itemize}
  We can therefore answer queries as follows.
  If $|K| > 4k^2 + 2k$ or $|S| > 2k$ we return that $G$ has no edge dominating set of size at most $k$.
  Otherwise, we solve edge dominating set on the subgraph of $G$ induced by $K \cup S$, which has at most $4k^2 + 4k$ vertices.
  This can be done using an algorithm due to Iwaide and Nagamochi~\cite{IwaideNagamochi2016}, in $O(2.2351^k)$ time.
\end{proof}

\subsection{Point Line Cover}
\label{sec:pointlinecover}
\defparproblem{Point Line Cover}{$k$}{A set $\mathcal P$ of $n$ points in the plane and an integer $k\in\mathbb N$.}{Find a set of at most $k$ lines passing through all the points in $\mathcal P$.}

Such a set of lines covering all points in $\mathcal P$ is called a \emph{line cover} of $\mathcal P$.
We will also consider the \emph{dual} problem known as \emph{Line Point Cover}.

\defparproblem{Line Point Cover}{$k$}{A set $\mathcal L$ of $n$ lines in the plane and an integer $k\in\mathbb N$.}{Find a set of at most $k$ points passing through all the lines in $\mathcal L$.}

It is known from folklore that these problems are equivalent, by replacing the point $(a,b)$ with the line $y = ax - b$ or vice versa.
These problems fall into a more general class of geometric problems, best described as \emph{covering things with things}~\cite{LangermanMorin2005}, for which we believe it should typically be possible to find dynamic algorithms with small update time.

\begin{theorem}
  There is a dynamic algorithm for {\sc Point Line Cover} which handles edge insertions in $O(g(k)^2)$ time, edge deletions in $O(g(k)^3)$ time, and queries in $O(g(k)^{2g(k)+2})$ time, under the promise that there is a computable function $g$ such that the point set can always be covered by at most $g(k)$ lines.
\end{theorem}
\begin{proof}
  If $g(k) < k$ then we are always in a ``no''-instance, so assume that $g(k) \geq k$.

  The main idea is to note that if there is a line cover with at most $k$ lines, then any line which passes through at least $k+1$ points must be contained in the line cover.
  In light of our promise, we will only use the weaker fact that any line which passes through at least $g(k)+1$ points must be contained in any line cover with at most $k$ lines.
  We will therefore maintain a set $\mathcal L_H$ of lines which pass through at least $g(k)+1$ points, and for each $\ell \in \mathcal L_H$, a set $\mathcal P_\ell$ of at least $g(k) + 1$ points on that line.
  We will also maintain the set~$\mathcal P'$ of points that are not in $\mathcal P_\ell$ for any $\ell \in \mathcal L_H$.
  We will further maintain that each point is in exactly one such set.

  Since every line in $\mathcal L_H$ must be in a line cover of size at most $g(k)$, our promise implies that we will always have $|\mathcal L_H| \leq g(k)$.
  Furthermore, since no line covers more than $g(k)$ points of $\mathcal P'$, we must always have $|\mathcal P'| \leq g(k)^2$.
  With these remarks, it is straightforward to maintain $\mathcal L_H$, $\mathcal P_\ell$, and~$\mathcal P'$ in $O(g(k)^2)$ time per insertion and $O(g(k)^3)$ time per deletion.
  When a new point $p$ is inserted, we first check for each line $\ell$ in $\mathcal L_H$ whether $p$ is on $\ell$, and if so we add it to $\mathcal P_\ell$ and conclude.
  If it is not on any of these lines, we add it to $\mathcal P'$.
  Then, for each line formed by $p$ and another point in $\mathcal P'$, we check whether that line contains at least $g(k)+1$ points in $\mathcal P'$.
  If we find such a line $\ell'$, we add~$\ell'$ to $\mathcal L_H$, and we remove those $g(k)+1$ points from $\mathcal P'$ and add them instead to $\mathcal P_{\ell'}$.
  When a point~$p$ is removed, we remove it from $\mathcal P'$ if it is in that set.
  If instead it is in $\mathcal P_\ell$ for some line $\ell \in \mathcal L_H$, then we remove it from $\mathcal P_\ell$.
  If $\mathcal P_\ell$ now consists of at most $g(k)$ points, then we remove $\ell$ from $\mathcal L_H$, and reinsert all the points from $\mathcal P_\ell$ as above.

  Finally, we must describe how to go from the sets we maintain to the point line cover of size at most~$k$ (or the conclusion that none currently exists) in order to answer queries.
  First, every line in $\mathcal L_H$ must be included since these lines each contain at least $g(k)+1 \geq k+1$ points.
  If there are more than $k$ such lines, we return that there is no line cover of size at most $k$.
  Otherwise, if there are $a \leq k$ such lines, we need to determine if there is a line cover of $\mathcal P'$ with only $k-a$ lines.
  Since $|\mathcal P'| \leq g(k)^2$, this can be solved in $O(g(k)^{2g(k)+2})$ time, using a simple static branching algorithm such as the one by Langerman and Morin~\cite[Theorem 1]{LangermanMorin2005}.
  If we find such a set $S$ of lines, we return $S \cup \mathcal L_H$, and otherwise we return that there is no line cover of $\mathcal P$ with size at most~$k$.
\end{proof}

\subsection{\texorpdfstring{$d$}{d}-Hitting Set}
\defparproblem{$d$-Hitting Set}{$k,d$}{A universe $U$ and a family $\mathcal F$ of subsets of $U$, each of cardinality exactly $d$.}{Find a set $X\subseteq U$ of at most $k$ elements that intersects all sets in $\mathcal F$.}

This problem generalizes {\sc Vertex Cover}, for which $d = 2$.

In the dynamic model, similar to {\sc Vertex Cover}, subsets of $U$ are inserted and deleted to/from~$\mathcal F$.
It is not hard to generalize the branching tree based dynamic algorithm to obtain a randomized algorithm with expected worst case update time $O(kd^k)$ and query time $O(k)$. (For details, see the end of this subsection.)

Here we will obtain a dynamic kernel based algorithm 
that supports updates in $g(k,d)$ time for some function $g$ that we very loosely bound by $g(k,d)\leq (d!)^d k^{O(d^2)}$.
It easily supports queries in $O(d^{k} d! (k+1)^d)$ time by running a branching algorithm on the kernel. 

Our dynamic kernel is a non-trivial and tricky generalization of the kernel for {\sc Vertex Cover}.
In the process we obtain a novel static kernel for the problem.
Interestingly, we improve (in the dependence on~$d$) upon the previous best linear-time kernel by van Bevern~\cite{vanBevern2014} by obtaining a kernel on $(d-1)! k(k+1)^{d-1}$ sets and $d! k(k+1)^{d-1}$ elements.
Van Bevern's kernel has $d! d^{d+1} (k+1)^d$ sets, and hence we save a factor of ~$d^{d+2} (1+1/k)$.
It is known that a polynomial time constructible kernel of size $O(k^{d-\eps})$ for any constant $\eps>0$ would imply that $\mathsf{NP}\subseteq \mathsf{coNP}/poly$~\cite{DellvanMelkebeek2014}.
Thus, the $k^d$ dependence on the number of sets in the kernel is optimal, assuming $\mathsf{NP}\not\subseteq \mathsf{coNP}/poly$.

\subsubsection{Algorithm Outline}
Before going into the details of the kernel and how to dynamically maintain it, we outline the main ideas behind the algorithm.
We first define the notion of an ``$(\ell, r)$-good subset'' of our universe $U$ of elements. 
A consequence of our formal definition is that a set $S \subseteq U$ is $(\ell, r)$-good if $|S| = \ell$, and there are many ways to add $r$ elements to it to obtain a set $S'$ which is itself good, and hence must be hit by any $d$-hitting set of $(U,\mathcal F)$ with size at most $k$.
This notion is defined recursively: whether or not a subset is $(\ell, r)$-good depends on whether some other related subsets are $(a,b)$-good for either $a > \ell$, or $a = \ell$ and $b < r$.
As a base case, all sets in $\mathcal F$ are good.
Our kernel will correspond to the set of minimal good sets.
In Lemma \ref{lem:dhitkernel} we show that this is a valid kernel, and in Lemma \ref{lem:dhitsmall} we show that, when there is a $d$-hitting set of $(U,\mathcal F)$ with size at most $k$, then the kernel is bounded in size by a function of $k$ and $d$.

We now outline how to dynamically maintain this kernel.
When a set $S$ is added to or removed from~$\mathcal F$, it can toggle which other sets are good. 
There are four ways in which this can happen, which are handled by four subroutines of our algorithm:
\begin{itemize}[labelindent=2em,labelwidth=\widthof{UPSTRONG},itemindent=3.5em,leftmargin=1.5em]
  \item[\texttt{DOWNINS}] If $S$ becomes good, then it might cause subsets of $S$ to become good as well.
    We can iterate over all subsets of $S$ and check whether the conditions have now been met.
  \item[\texttt{UPWEAK}] If $S$ becomes good, then if any supersets of $S$ were good and in the kernel, we need to remove them from the kernel since they are no longer minimal good sets.
    Since $S$ was not good before the update, there are not too many such sets, and we can find them efficiently using some pointers we maintain throughout the algorithm.
  \item[\texttt{DOWNDEL}] If $S$ becomes no longer good, then some of its subsets might also become no longer good, similar to \texttt{DOWNINS}.
  \item[\texttt{UPSTRONG}] If $S$ becomes no longer good, then its good supersets might become minimal good sets, similar to \texttt{UPWEAK}.
\end{itemize}
Notably, each of these procedures can change whether subsets other than $S$ are good, which can in turn trigger more procedures.
The recursive nature of our definition of a good set guarantees that this process can only go on for $k \cdot d$ depth of recursion.

\subsubsection{Static Kernel}
We now begin by formally describing the static kernel.
Let $\mathcal F$ be the given family of sets and let $U$ be the universe.
We will recursively define the notion of a {\em good} set.
\begin{definition}[good set]
  Let $r\in\mathbb N$ and $\nu_r=r! (k+1)^r$.
  Let $d\in\mathbb N$ and let $(U,\mathcal F)$ be an instance of {\sc $d$-Hitting Set}.
  We define the notion of an ``$(\ell,r)$-good'' set inductively, in decreasing order of~$\ell$ from $d$ to~$1$, and for fixed $\ell$, for increasing $r$ from $1$ to $d-\ell$.
  \begin{itemize}
    \item Any set $S\in \mathcal F$ is \emph{$(d,r)$-good} for all $r$.
    \item A set $S\subseteq U$ is \emph{$\ell$-good} if $S$ is $(\ell,r)$-good for some $r$.
    \item A set $S\subseteq U$ is \emph{$(\ell',r)$-strong} if $S$ is $\ell'$-good and does not contain any $(\ell'-j,j)$-good subsets for any $j\in\{1,\hdots,r-1\}$.
   \item A set $S\subseteq U$ is \emph{$(\ell,r)$-good} if $S$ is of size $\ell$ and is contained in at least $\nu_r$ $(\ell+r,r)$-strong sets.
   \item A set $S\subseteq U$ is \emph{good} if $S$ is $\ell$-good for $\ell = |S|$.
  \end{itemize}

\end{definition}
Notice that if a set is $(\ell',r)$-strong, then it is also $(\ell',r')$-strong for all $r'<r$. Also, any $\ell$-good set is $(\ell,1)$-strong.
Further, note that since the notion of $(\ell+r,r)$-strong only depends on $(\ell+a,r-a)$-good sets for $a\geq 1$, the definition of $(\ell,r)$-good is sound.

We will show the following useful lemma:
\begin{lemma}
\label{lem:dhitkernel}
  Let $(U,\mathcal F)$ be an instance of {\sc $d$-Hitting Set}.
  If $(U,\mathcal F)$ admits a hitting set $X$ of size at most~$k$, then any good set $S\subseteq U$ intersects $X$.
\end{lemma}
\begin{proof}
  We will prove the statement of the lemma by induction.
  The base case is that $S\in \mathcal F$.
  Then, by definition of a hitting set, $S\cap X\neq \emptyset$.
  For the inductive case suppose that every $\ell'$-good set for $\ell'>\ell$ contains an element in $X$ and let $S$ be an $(\ell,j)$-good set.

  By definition, $S$ is of size $\ell$ and is contained in at least $\nu_r$ $(\ell+r)$-good sets that do not contain any $(\ell+j,r-j)$-good subsets for any $j\in\{1,\hdots,r-1\}$.
  Let the $(\ell+r)$-good sets be $X_1,\ldots,X_L$.
  Consider~$X_1$, and how many $X_j$ intersect it in more than $S$.
  Let $x\in X_1\setminus S$.
  Then $|S\cup \{x\}|=\ell+1$, and $X_1$ does not contain $(\ell+1,r-1)$-good sets.
  Thus, $S\cup \{x\}$ must lie in fewer than $n_{r-1}$ $(\ell+r)$-good sets that do not contain $(\ell+j+1,r-1-j)$-good subsets for any $j\in\{1,\hdots,r-2\}$.
  This also means that $S\cup \{x\}$ is contained in fewer than $\nu_{r-1}$ of the sets $X_1,\ldots,X_L$.
  Thus, if we remove all sets $X_j$ (for $j>1$) that intersect $S\cup \{x\}$ for each $x\in X_1\setminus S$, we have removed fewer than $r\nu_{r-1}$ sets. 
  Since $L\geq \nu_r=r(k+1)\nu_{r-1}$, we can repeat this greedy procedure at least $k+1$ times. The collection of the at least $k+1$ sets $X_1$ only intersect in $S$ and hence $S\cap X\neq \emptyset$ as otherwise $X$ must contain at least one element from each of $k+1$ disjoint sets.
\end{proof}

Let $\mathcal F'$ consist of those sets $S\subseteq U$ that are good and none of their subsets are good.
Let~$U'$ consist of all $u\in U$ that are contained in some set of $\mathcal F'$.
The above lemma states that $(K,k)$ with $K=(U',\mathcal F')$ is a kernel for the instance $(U,\mathcal F)$.
First, if $X'$ is a hitting set of $K$, it must be a hitting set for $\mathcal F$ as well since for every $F\in \mathcal F$, either $F\in \mathcal F'$ or some subset of $F$ is in $\mathcal F'$.
Furthermore, let $X$ be a hitting set of~$\mathcal F$ with size at most $k$.
By the lemma, if some $S$ is in $\mathcal F'$, then it must intersect~$X$ non-trivially and so~$X$ is a hitting set of~$\mathcal F'$ as well.

Now we argue about the size of $K$.
\begin{lemma}
\label{lem:dhitsmall}
  If $(U,\mathcal F)$ (and hence also $(U',\mathcal F')$) admits a hitting set of size at most $k$, then $|U'|\leq d|\mathcal F'|$ and $|\mathcal F'|\leq (1+\frac{2}{(k+1)(d-1)})\cdot d! (k+1)^{d}$.
\end{lemma}
\begin{proof}
  If $\{u\}\in\mathcal F'$, then no other set containing $u$ can be in~$\mathcal F'$. 
  Otherwise, consider all sets of size $r+1$ in~$\mathcal F'$ that contain $u$, for any choice of $r\in\{1,\hdots,d-1\}$.

  Since $\{u\}\notin \mathcal F'$, we know that $u$ cannot be $(1,r)$-good, and thus $u$ is contained in fewer than $\nu_r$ $(r+1)$-good sets that do not contain any $(j+1,r-j)$-good subsets for any $j\in\{2,\hdots,r-1\}$.
  Now since for every $F\in \mathcal F'$ we have that it contains no good subsets, this means that $u$ is contained in fewer than~$\nu_r$ sets in $\mathcal F'$ of size $r+1$.

  Thus, the number of sets of $\mathcal F'$ containing $u$ is at most 
  \begin{equation*}
    \sum_{r=1}^{d-1}\nu_r    = \sum_{r=1}^{d-1} r! (k+1)^r
                          \leq \left(1+\frac{2}{(k+1)(d-1)}\right) (d-1)!(k+1)^{d-1},
  \end{equation*}
  where the last inequality can be proven inductively.
  Thus, if there is a hitting set of size at most $k$ for $\mathcal F'$, then the size of $\mathcal F'$ is at most $(1+\frac{2}{(k+1)(d-1)})d! (k+1)^d$.
\end{proof}

Before we discuss how to dynamically maintain this kernel, we remark that it can be statically computed efficiently.

\begin{lemma}
  The static kernel can be computed in $O(3^d n+m)$ time.
\end{lemma}
\begin{proof}
  We will initialize a list $R$ to contain all sets in $\mathcal F$.
  Then we will iterate $d$ times, once for each size~$i$ of a set.
  In iteration $i$ we have the current list $R$ of sets of size $\geq i$ such that (1) they are all good and (2) if their size is $j>i$, then they have no $(t,j-t)$-good subsets for any $i\leq t<j$, and are thus $(j,j-i)$-strong.

  Consider each $S\in R$, and let $j = |S|$.
  For every subset $T\subseteq S$ of size $i$, increment the counter $c(R,j-i)$ that determines whether $T$ needs to become $(i,j-i)$-good.
  After this, go through all touched subsets $T$ of size $i$ and mark each $T$ whose counter became big enough to make it $(i,j-i)$-good, also inserting $T$ into $R$.
  After this, iterate through $R$ again and for every set $S$ of size $>i$, check if it has a subset of size $i$ that is now marked and hence good; if so, remove $S$ from $R$.

  This maintains the invariant that at each iteration, the sets in $R$ are either of size $i$, or are of size $j > i$ and are $(j,j-i)$-strong.
  Thus at the end of all iterations, we will only have minimally good sets. 
  Also by induction we can show that for any $(i,j-i)$-good set $X$ of size $i$, in iteration $i$ all its $(j,j-i)$-strong supersets (for any $j$) would be in $R$ and hence~$X$ will be added to~$R$ and all its supersets removed.
  Hence~$R$ will consists of all minimal good sets and no other sets.

  The total runtime is $O(m +n\sum_{i=1}^d \binom{d}{i} 2^i)\leq O(m+ n3^d)$.
  This is because, for every $i$, every set $A\in\mathcal F$, and every subset $B\subseteq A$ of size $i$, over all iterations in which $B$ is in $R$, we iterate over at most all of its subsets.
  Moreover, in different iterations, we iterate over subsets of different sizes, so that no subset is considered more than once as a subset of $B$.
\end{proof}

\subsubsection{Dynamically Maintaining the Kernel}
To make the kernel dynamic, we need to be able to maintain the subsets $S$ that are $(\ell,r)$-good and those that are inclusion minimally good.

First, we keep for every set $S\subseteq U$ of size $\ell$ a list $L_{S,r}$ of pointers to all $(\ell+r,r)$-strong sets that contain~$S$.
We also keep $c(S,r)=|L_{S,r}|$.
Note that if $c_{S,r}\geq \nu_r$, then $S$ is $(\ell,r)$-good.

We will first show how to update $\mathcal F'$, given that we know for every set for which $r$ it is $(\ell,r)$-good. 

Consider the set $W$ of sets that the goodness update algorithm updates.
Their number is $g(k,d)$ for some function $g$ that we will bound later.
Consider each set $S\in W$ in {\em order of non-decreasing size}. 

Suppose first that $S$ is good; it might have been good or not before the goodness update. Consider all its (at most $2^d$) subsets: if $S$ has no subset that is good, then we add it to $\mathcal F'$, and if it has a good subset, then we remove it from $\mathcal F'$. 

Now, if $S$ used to be not good but the updates made it good, then we need to make sure that none of its supersets are in $\mathcal F'$.
Note that if $S$ has a good subset, then all its supersets have been handled by handling one if its newly good subsets, if any.
Hence, we can assume that $S$ has no good subsets and has been added to $\mathcal F'$.

We only care about those supersets $T$ of $S$ that used to be in $\mathcal F'$ and hence were good-minimal.
Every such $T$ was by definition $(|T|,r)$-strong for all $r$.

However, since $S$ was not good before the update, then for every $r$ it had $<\nu_r$ $(|S|+r,r)$-strong supersets, and in particular $T$ was in $L_{S,|T|-|S|}$ before the goodness update.
The goodness algorithm then should leave behind old copies of the lists of any modified set, and then we can use these lists to access all relevant supersets of $S$ in $O(\sum_r \nu_r)\leq O(d! (k+1)^d)$ time and remove them from ${\mathcal F'}$.

Now suppose that $S$ was updated and is not good.
If $S$ was not good before the update, it was not in~$\mathcal F$ so we do not need to do anything. 
Hence assume that $S$ was good before the update and now is not.
If $S$ has some good subset $T$, then we do not need to do anything since either $T$ was good to begin with and $S$ was not in $\mathcal F'$, or $T$ is newly good and then $T$ was processed already (and hence all its supersets were fixed).
Suppose that $S$ has no good subsets.
If one of them is newly not good, we have fixed $S$ already.
Thus, all subsets of $S$ were not good to begin with.
Suppose that $S$ was in $\mathcal F'$.
Now we remove it.
We need to find all inclusion minimal supersets of $S$ that are good.
Since $S$ is currently not $(\ell,r)$-good for any $r$ and all supersets we care about are $r$-strong for all $r$ due to minimality, the number of supersets of $S$ that we care about is at most $\sum_r \nu_r\leq d d! (k+1)^d$.
We go through all of these (via $L(S,\cdot)$) and add them to $\mathcal F'$ if they have no good subsets.

Thus, if the goodness updates can run in $g(k,d)$ time for some function $g$, then the entire update procedure can run in $g'(k,d)$ time for some function $g'$.

Now we show how to update the $(\ell,r)$-goodness information. 
For this task, we have four procedures, that we now describe in detail.

\medskip
\noindent
\textbf{Procedure $\mathtt{DOWNINS}(d',Q)$.}
This procedure is given a set $Q$ of size $d'$ that has newly become $(d',r)$-good for some $r$.
The goal of the procedure is to consider all subsets of $Q$, check whether they have become $(d'-j,j)$-good and we need to fix their status and propagate the changes to other affected subsets.
An invariant of this procedure is that any set affected by it is of size at most $d'$.

For $i$ from $1$ down to $d'-1$ we consider the subsets $T$ of $Q$ of size $d'-i$.
Say we are at~$i$: we check whether~$Q$ is $(d',i)$-strong and if so, we add it to $L_{T,i}$ for every $T\subseteq Q$ of size $d'-i$, updating the counters~$c(T,i)$.
Some sets $T$ might now become $(d'-i,i)$-good.

First notice that if some $T$ became $(d'-i,i)$-good for some $i$, no subset of $Q$ can become $(d'-j,j)$-good for any $j>i$ since $Q$ will not be $(d',j)$-strong for any $j>i$ since it contains a set that is $(d'-i,i)$-good.
Thus, there is at most one step $i$ in which sets $T$ become $(d'-i,i)$-good.
Updating lists and counters for the subsets of $Q$ takes $O(2^{d'})$ time.

Now for (at most 1) $i$, some sets $T$ become $(d'-i,i)$-good.
Two steps need to happen for every one of the at most ${\binom{d'}{i}}$ such $T$. 
The first step happens if $i$ is the only value for which~$T$ is good.
Then $T$ was not good at all before the update.
Then some of its subsets might need to become good and propagate. 
This is just a call to $\mathtt{DOWNINS}(d'-i,T)$.
Note that this call only touches sets of size at most $d'-i$.

The second step is required because some supersets of $T$ may need to update their strongness information since now they contain a $(d'-i,i)$-good set.
The only supersets affected are those of size~$d'$ and that are $(d',j)$-strong for some $j>i$, from the definition of strongness and now might only be $(d',i)$-strong.
Moreover, we know that since $T$ was not $(d'-i,i)$-good before, it had at most $\nu_i-1$ supersets that can be $(d',i)$-strong; together with $Q$, these are all the $(d',i)$-strong sets in $L(T,i)$.
Thus it can also have at most $\nu_i-1$ supersets~$Q'$ that are $(d',j)$-strong since for $j>i$, any $(d',j)$-strong set is $(d',i)$-strong and these sets can all be found in $L(T,i)$.
We thus call $\mathtt{UPWEAK}(d',i)$ on all supersets $Q'\in L(T,i)$ that are $(d',j)$ strong since for some $j>i$.
These calls only touch sets of size $d'$ and size at most $d'-i-1$.
Thus, the sets of size $d'-i$ will be finished before the $\mathtt{UPWEAK}$ calls.

Thus for some $i>0$, the runtime is
\begin{equation*}
  T(\mathtt{DOWNINS}(d')) = O(2^{d'}) + {\binom{d'}{i}}\cdot [T(\mathtt{DOWNINS}(d'-i)) + \nu_i\cdot T(\mathtt{UPWEAK}(d',i))] \enspace .
\end{equation*}

\medskip
\noindent
\textbf{Procedure $\mathtt{UPWEAK}(d',i,S)$.}
Here we are given a set $S$ that is $(d',j)$-strong for some $j>i$ but now needs to be made no longer $(d',j)$-strong for all $j>i$.
We first set a flag that $S$ is no longer $(d',j)$-strong.
We then access all subsets $T$ of $S$ of size $d-j$ for each choice of $j>i$; there are at most $2^{d'}$ of these. 
We remove $S$ from $L_{T,j}$, decrementing the counter $c(T,j)$.
Now $T$ might not be $(d'-j,j)$-good anymore, so we need to propagate this information to all their subsets and supersets.

Note that the subsets only need to change their information if $T$ is not $(d'-j,r)$ good for any $r$.
We detect when this happens using our counters $c(T,r)$, and if it does, we propagate the information to the subsets using a call to $\mathtt{DOWNDEL}(d'-j)$.
Note that this propagation accesses only sets of size at most $d'-j$.

We propagate the information to the supersets of $T$ as follows:
We need to fix any supersets of size~$d'$ that might need to become $(d',k)$-strong for $k>j$ due to making $T$ no longer $(d'-j,j)$-good.
We only need to look at those supersets of $T$ that are already $(d,j)$-strong, since a $(d',k)$-strong set for $k>j$ must also be $(d',j)$-strong.
The number of such supersets of $T$ is smaller than $\nu_j$, since $T$ is not $(d'-j,j)$-good.
We call $\mathtt{UPSTRONG}(d',j, S')$ on every such superset $S'$ that is $(d',j)$-strong (note that $S$ will not be touched).
This procedure only touches sets of size $d'$ or of size at most $d'-j-1$.

Thus overall, the call to $\mathtt{UPWEAK}$ only touches sets of size $d'$ or of size at most $d'-i-1$.
The runtime is $T(\mathtt{UPWEAK}(d',i))\leq \sum_{j>i} {\binom{d'}{j}}\cdot [T(\mathtt{DOWNDEL}(d'-j)) + \nu_j\cdot T(UPSTRONG(d',j))]$.
A loose upper bound on this runtime is
\begin{equation*}
  T(\mathtt{UPWEAK}(d',i))\leq 2^{d'}\cdot [T(\mathtt{DOWNDEL}(d'-i-1)) + \nu_{d'}\cdot T(\mathtt{UPSTRONG}(d',i+1))] \enspace .
\end{equation*}

\medskip
\noindent
\textbf{Procedure $\mathtt{DOWNDEL(d',Q)}$.}
Here one is given a set $Q$ of size $d'$ that has become no longer good after an update.
The goal is to propagate the information to all subsets of $Q$ and if they are changed, then to their supersets and so on.
The approach is symmetric to that of $\mathtt{DOWNINS}$, and the same invariant is maintained: the procedure only touches sets of size at most $d'$.

For $i$ from $1$ down to $d'-1$, we consider the subsets $T$ of $Q$ of size $d'-i$.
Say we are at $i$: we check whether $Q$ was in $L_{T,i}$ and if so, remove it and decrement $c(T,i)$.
If now $c(T,i)<\nu_i$ holds, then $T$ is no longer $(d'-i,i)$-good, and we will need to propagate this information to its subsets and supersets.

First notice that if some $T$ just became no longer $(d'-i,i)$-good for some $i$, then no subset of $Q$ was $(d'-j,j)$-good for any $j>i$, as before the update $Q$ was not $(d',j)$-strong for any $j>i$ as it contained a $(d'-i,i)$-good set.
% so no subset of $T$ can now be $(d-j,j)$-good.
Thus, there is at most one step $i$ in which sets $T$ become no longer $(d'-i,i)$-good.
Updating lists and counters for the subsets of $Q$ takes $O(2^{d'})$ time.
%***might want to also keep whether $Q$ is strong***

Now for (at most 1) $i$, some sets $T$ become no longer $(d'-i,i)$-good.
Two steps need to happen for every one of the at most ${\binom{d'}{i}}$ such $T$. 
The first step happens if $T$ is not good at all anymore (and so $T$ was $(d'-i,i)$-good only for $i$).
Then some of its subsets may no longer be good, and we propagate this information.
This amounts to a call to $\mathtt{DOWNDEL}(d'-i,T)$.
Note that this call only touches sets of size at most $d'-i$.

The second step happens as some supersets of $T$ may need to update their strongness information, since now the number of their $(d'-i,i)$-good subsets drops.
If this number drops to $0$, then they may become $(d-i',j)$-strong for some $j>i$. 
The only supersets affected are those of size $d'$ and that are already $(d',i)$-strong. 
Moreover, we know that since~$T$ is not $(d'-i,i)$-good, it has at most $\nu_i-1$ supersets that can be $(d',i)$-strong, and all of them are in $L(T,i)$.
We accomplish this propagation of information by calling $\mathtt{UPSTRONG}(d',i)$ on all supersets $Q'\in L(T,i)$.
These calls only touch sets of size $d'$ and size at most $d'-i$.
Thus, handling sets of size $d'-i$ will be finished before the $\mathtt{UPSTRONG}$ calls.

Thus for some $i>0$, the runtime is bounded by $T(\mathtt{DOWNDEL}(d'))\leq O(2^{d'}) + {\binom{d'}{i}}\cdot [T(\mathtt{DOWNDEL}(d'-i)) + \nu_i\cdot T(\mathtt{UPSTRONG}(d',i))]$.
A loose bound on this expression is $T(\mathtt{DOWNDEL}(d'))\leq 2^{d'}\cdot [T(\mathtt{DOWNDEL}(d'-1)) + \nu_{d'}\cdot T(\mathtt{UPSTRONG}(d',1))]$.

\medskip
\noindent
\textbf{Procedure $\mathtt{UPSTRONG}(d',i,S)$.}
Here we are given a set $S$ that is $(d',i)$-strong but the number of its $(d'-i,i)$-good subsets has dropped to $0$, and we need to make $S$ now $(d',i+1)$-strong and maybe also $(d',j)$-strong for other values $j>i$.

We start with $j=i+1$ and then access all subsets $T$ of $S$ of size $d'-j$; there are at most~$2^{d'}$ of these over all $j$. 
We add $S$ to $L_{T,j}$, incrementing the counter $c(T,j)$.
If none of these sets $T$ becomes $(d'-j,j)$-good, we conclude that $S$ needs to also be $(d',j+1)$-strong, and so we increment~$j$ and loop again.

Now suppose that we reach some $j>i$ such that some set $T$ has become $(d'-j,j)$-good (if $c(T,j)\geq \nu_{j}$), so we need to propagate this information to all the subsets and supersets of~$T$.
Then we see that we can stop incrementing $j$, since $S$ is not $(d',\ell)$-strong for any $\ell>j$.
Hence, we only need to deal with the newly $(d'-j,j)$-good subsets.

Fix any $T$ that became $(d'-j,j)$-good.
Note that the subsets of $T$ only need to change their information if $T$ was not good at all before it just became $(d'-j,j)$-good.
We detect when this happens using our counters $c(T,r)$, and if it does, we propagate the information to the subsets using a call to $\mathtt{DOWNINS}(d'-j)$.
Note that this propagation accesses only sets of size at most $d'-j$.

We propagate the information to the supersets of $T$ as follows:
We need to fix any supersets of size~$d'$ that might need to stop being $(d',\ell)$-strong for $\ell>j$ due to making $T$ $(d'-j,j)$-good.
We only need to look at those supersets of $T$ that are already $(d,j)$-strong, since a $(d',\ell)$-strong set for $\ell>j$ must also be $(d',j)$-strong.
The number of such supersets of $T$ is smaller than $\nu_j$, since $T$ was not $(d'-j,j)$-good before the update, and we can access all these sets through $L(T,j)$.
We call $\mathtt{UPWEAK}(d',j, S')$ on every such superset~$S'$ that is $(d',j)$-strong.
This procedure only touches sets of size $d'$ and of size at most $d'-j-1$.

Thus overall, the call to $\mathtt{UPSTRONG}$ only touches sets of size $d'$ or of size at most $d'-i-1$.
The runtime is $T(\mathtt{UPSTRONG}(d',i))\leq \sum_{j>i} {\binom{d'}{j}}\cdot [T(\mathtt{DOWNINS}(d'-j)) + \nu_j\cdot T(\mathtt{UPWEAK}(d',j))]$.
A loose upper bound on this expression is
\begin{equation*}
  T(\mathtt{UPSTRONG}(d',i))\leq 2^{d'}\cdot [T(\mathtt{DOWNINS}(d'-i-1)) + \nu_{d'}\cdot T(\mathtt{UPWEAK}(d',i+1))] \enspace .
\end{equation*}

This completes the description of the procedures.

Using induction, one can prove the following somewhat loose bounds on their run time:
\begin{eqnarray*}
  T(\mathtt{UPSTRONG}(d',i)),T(\mathtt{UPWEAK}(d',i)) & \leq & O(2^{(d'+1)(d'-i)}\nu_{d'}^{d'-i}),\\
     T(\mathtt{DOWNINS}(d')), T(\mathtt{DOWNDEL}(d')) & \leq & O(2^{(d'+1)d'} \nu_{d'}^{d'}) \enspace .
\end{eqnarray*}

Suppose that a set $Q$ of size $d$ is inserted into $\mathcal F$.
Then we simply call $\mathtt{DOWNINS}(d,Q)$.
If~$Q$ is deleted, we call $\mathtt{DOWNDEL}(d,Q)$.
These runtimes subsume the time to update $\mathcal F'$, so the total update runtime is at most $O(2^{d(d+1)}\nu_d^d) = O(2^{d(d+1)} (d!)^d (k+1)^{d^2}) = (d!)^d k^{O(d^2)}$.

We point out that we did not try to optimize the update time bound here, only striving to achieve a function of $k$ and $d$.

\subsubsection{Dynamic branching tree}
For completeness, we include the generalization of the dynamic branching tree for {\sc Vertex Cover} to {\sc $d$-Hitting Set}.

Given an instance $(\mathcal F,U)$ of {\sc $d$-Hitting Set}, we will keep a binary branching tree $\mathcal T$ whose every non-leaf node contains a set in $\mathcal F$ to branch on.
Every non-leaf subtree $\mathcal T'$ represents an instance $(\mathcal F',U)$ of {\sc $d$-Hitting Set}, where $\mathcal F'\subseteq \mathcal F$.

If a tree node $t$ contains a set $F = (u_1,\hdots,u_{d'})\in\mathcal F_t$ (for $d'\leq d$) and the subtree $\mathcal{T}_t$ rooted at $t$ corresponds to $\mathcal F'$, then the $i$-th subtree of $t$ corresponds to putting $u_i$ into the hitting set $X$ and contains the sets in $\mathcal F_t$ that are \emph{not} hit by $u_i$, for $i = 1,\hdots,d'$.

We will maintain the invariant that every non-leaf node $v$ of $\mathcal T$ contains an edge chosen uniformly at random among the sets in the set $\mathcal F_t$ corresponding to the subtree under $v$.

A node $t$ of $\mathcal T$ becomes a leaf if its subset $\mathcal F_t$ is empty, in which case it is a ``yes''-node, or if it is at depth~$k$, so that $k$ elements from $U$ have been added to the hitting set along the path from the root to it, and then it is a ``no''-node if the set $\mathcal F_t$ is non-empty and a ``yes''-node otherwise.

Suppose that a new set $F$ in inserted.
To update $\mathcal T$, we traverse each node of $\mathcal T$ starting at the root, as follows.
Suppose that we are at a node $v$ such that its corresponding instance $(\mathcal F_t,U)$ contains exactly $x_t = |\mathcal F_t|$ sets.
Then after adding the set $F$ increases the number of sets to $x_t + 1$.
We need to maintain the invariant that the sets in the nodes of $T$ are random in their subtrees.
Hence, with probability $1/(x_t + 1)$  we replace the set in $v$ with $F$ and rebuild the entire subtree of $T$ under~$v$.
If $F$ is not added into $v$, then the set $F$ in $v$ has probability $(1 - 1/(x_t + 1)) \cdot (1/x_t) = 1/(x_t + 1)$ of being in $v$, and so the invariant is maintained.
How long does this take?
If $v$ is at distance $i$ from the root, then the number of nodes under it is at most $ d^{k-i}$ since there are only $k-i$ branching decisions left.
At each such node $v$ at distance $i$ from the root, a random set is placed and for each of the two children, at most $x_t$ sets need to be looked at to check whether they are covered.
This takes no more than $O(x_td^{k-i})$ time to rebuild the subtree.
However, this only happens with probability $1/(x_t + 1)$, so in expectation the runtime is $O(d^{k-i})$.
There are $d^i$ nodes at distance $i$ from the root, so level $i$ of the tree is rebuilt in expected time $O(d^k)$ and the entire tree in expected time $O(kd^k)$.
This rebuilding ensures that each set in a subtree appears with the same probability at the root of the subtree maintaining our invariant.
To find a hitting set of size $k$, one only needs to traverse the at most~$d^k$ leaves of the tree and find a ``yes''-node if one exists (if it does not, just return ``no'' as there is no hitting set of size at most $k$).
After finding the ``yes''-node, just go up the tree to the root to find the elements chosen to be in the hitting set.
In fact, we can improve the query time by maintaining in addition a pointer to a ``yes''-leaf node, if one exists.
Then the query time is $O(k)$.

Now suppose that a set $F$ is deleted.
Suppose first that $F$ does not appear anywhere in $\mathcal T$.
Then consider any subtree of $T$ under some tree node $v$ such that $F$ is in the instance $(\mathcal F_v,U)$ corresponding to the subtree.
The set $F'$ stored in $v$ is chosen uniformly at random from the sets of~$\mathcal F_t$.
When we remove~$F$, set $F'$ is still uniformly random among the remaining sets of~$\mathcal F_t$, so we do not have to do anything.
Now assume that $F$ does appear in $\mathcal T$.
At each node~$v$ of $\mathcal T$ containing $F$, we need to remove $F$ from $v$ and rebuild the entire subtree under~$v$.
We will do this top to bottom starting at the root, searching for tree nodes containing~$F$ in, say, a breadth first search fashion.
Consider one such $v$ containing $F$.
Suppose that the subinstance $\mathcal F_t$ corresponding to it contains~$x_t$ sets and that $v$ is at distance $i$ from the root.
Similar to before, we remove $F$ and rebuild the subtree in time $O(x_td^{k-i})$.
However,~$F$ only had a $1/x_t$ chance of being in node $v$ to begin with, so that we only spend
$O(d^{k-i})$ time rebuilding that node in expectation.
Summing over all nodes we still get an expected time of $O(kd^k)$.
We end up with a dynamic algorithm for {\sc $d$-Hitting Set} with expected worst case update time $O(kd^k)$ and $O(k)$ query time.

\section{Max Leaf Spanning Tree}
\defparproblem{Undirected Max Leaf Spanning Tree}{$k$}{A graph $G$ and an integer $k\in\mathbb N.$}{Does $G$ admit a spanning tree with at least $k$ leaves?}

We will use the dynamic kernel approach.
We will build on an extremely simple kernel with $O(k^2)$ vertices that is obtained by contracting nodes of degree $2$; this idea goes back to a kernelization algorithm by Cai et al.~\cite{CaiEtAl1997}.

We say a vertex $v$ in $G$ is \emph{useless} if it has degree two, and its neighbors each have degree two.
We \emph{resolve} a useless vertex $v$ by connecting its neighbors by an edge, and then deleting~$v$.
The \emph{resolved graph} of $G$, denoted~$G^*$, is the result of taking $G$ and repeatedly resolving useless vertices until no more remain.

The kernelization algorithm (for the decision version of the problem) works as follows:
\begin{enumerate}
  \item If $G$ is not connected, return a ``no''-instance of constant size.
  \item If $G$ has a node of degree at least $k$, return a ``yes''-instance of constant size.
  \item Compute the resolved graph $G^{*}$.
    If $G^{*}$ has at least $4k^2 + 12k + 8$ nodes then return a ``yes''-instance of constant size, otherwise return $G^{*}$.
\end{enumerate}

The resulting graph has $O(k^2)$ nodes and $O(k^3)$ edges.
Correctness is implied from the following result by Cai et al.~\cite{CaiEtAl1997}.
\begin{proposition}[{\cite[Theorem 2.6]{CaiEtAl1997}}]
  \label{MLSTKernel}
  Let $H$ be a connected graph without useless vertices and with maximum degree less than $k$.
  If $H$ has at least $4k^2 + 12k + 8$ vertices, then $H$ has a spanning tree with at least $k$ leaves.
\end{proposition}

We now show how to dynamically maintain this kernel, in such a way that we actually maintain a spanning tree, rather than just solving the decision problem. This data structure can be seen as a generalization of the past work on dynamically maintaining a spanning forest, which we discuss in Section \ref{sec:dynamictree}.
Since we are maintaining the spanning tree, we can answer simple queries like whether a given edge is in the tree, or what the $i$-th edge in the tree out of a given vertex is, in constant time.

\begin{theorem} 
  There is a data structure that for graphs $G$ and integers $k$, maintains a spanning forest $T$ of~$G$ such that
	\begin{itemize}
	\item if $G$ has a spanning tree with at least $k$ leaves, then $T$ is a spanning tree which has at least $k$ leaves,
	\item if $G$ does not have a spanning tree with at least $k$ leaves, then $T$ is an arbitrary spanning forest, and
	\item the data structure always knows which case it is in,
	\end{itemize}
	and that takes expected amortized $O(3.72^k + k^5 \log n + \log n \log\log^{O(1)} n)$ time or deterministic amortized $O(3.72^k + k^5 \log n + \log^2 n / \log \log n)$ time per edge insertion or deletion.
\end{theorem}
\begin{proof}
  We now provide the dynamic algorithm.
  We will maintain the resolved graph $G^{*}$ of our current graph~$G$, we will maintain a spanning forest $T^*$ (which does not necessarily have $k$ leaves) of the resolved graph~$G^{*}$, and we will maintain the spanning forest $T$ of $G$ corresponding to $T^*$.
  We also store a sorted list of all nodes in $G$ by degree, which can be maintained in constant time per update, and a dynamic tree structure~$D$ containing $T$, which can be maintained in amortized $O(\log n)$ time per update (see Section~\ref{sec:dynamictree}).
  Then, queries for a spanning tree with at least $k$ leaves will be answered in time $O(3.72^k + k^5 \log n)$ per update or query by a modification of $T$.
  We first describe how to maintain $G^{*}, T^*, T$, and $D$.

  Recall that $G^{*}$ is formed by contracting paths consisting only of useless vertices into a single edge.
  We will maintain $G^{*}$ along with a list $L_e$ on each edge $e$ corresponding to the path of useless nodes which have been contracted to that edge.
  The lists are stored as doubly-linked lists, with a pointer from any vertex in~$G$ to its position in a list in $G^{*}$.
  When an edge $e = \{u,v\}$ is inserted or deleted, we need to check whether~$u,v$, or one of their neighbors has changed whether it is a useless node.
  If not, we simply reflect the change in~$G^{*}$.
  We discuss how to deal with $u$; we then deal with $v$ similarly.
  If $u$ did not have degree 2 before or after the change then neither $u$ nor its neighbors changed whether they are useless.
  If $u$ now has degree 2, then either $u$ or its two neighbors may now be useless.
  We observe the degrees of those three nodes and, if they are degree two, their neighbors, to determine which of the three became useless.
  We reflect a new useless node $x$ in $G^{*}$ by removing $x$, and connecting its two neighbors $a,b$ with a new edge whose dynamic list is the concatenation $L_{(a,x)}, x, L_{(x,b)}$.
  If $u$ had degree 2 but now does not, we similarly check it and its former two neighbors to see whether they used to be useless but are now not.
  If a vertex $x$ becomes no longer useless, we can do the opposite of the above process to reflect this in~$G^{*}$: we find its position in a list following a pointer and split the list appropriately.
  Because we store pointers into the doubly linked lists we can split and merge lists in constant time.
  Thus, the entire update operation takes $O(1)$ time.

  Maintaining $T^*$ can be done in either expected amortized $O(\log n \log\log^{O(1)} n)$ time or deterministic amortized $O(\log^2 n / \log \log n)$ time per update to $G^*$, and hence per update to~$G$, depending on which data structure we select from Proposition~\ref{thm:dynamicconnectivity} in Section \ref{sec:dynamicconnectivity}.
  Now~$T$ corresponds almost directly to~$T^*$ except for edges of $G$ which have been compressed into lists in $G^{*}$.
  For each edge $e$ in $G^{*}$ with corresponding list $L(e)$ which contains more than one edge, we always include every edge in $L(e)$ except the final one in~$T$.
  If $e$ is included in $T^*$ then we also include the final edge of $L(e)$ in $T$.
  We make these corresponding constant-size changes to $T$ when the sets of useless vertices change in the previous paragraph.
  We also maintain a copy of $T$ in our dynamic tree structure $D$.
  Each time a change is made to~$T$, we make it in~$D$ as well.
  This will be used later to answer \texttt{after(a,b)} queries, where we want to find the vertex after $a$ on the path from $a$ to~$b$ in $T$, in $O(\log n)$ amortized time per query.

  We now describe how to query at a given time for a spanning tree with at least $k$ leaves.
  We can tell from our spanning forest $T$ whether $G$ is connected or not; if not, we return that there is no such spanning tree.
  Next, we check whether the highest degree vertex in our graph has degree at least~$k$ from our list of vertices sorted by degree.
  If so, let $v$ be that vertex, and pick any $k$ of its edges $e_1 = \{v,y_1\}, \ldots, e_k = \{v,y_k\}$.
  To find a tree with at least~$k$ leaves to return, we start by adding these $k$ edges into $T$. 
  For each $y_i$, we find the first edge out of $y_i$ on the path to $v$ in $T$ using $D$, and we remove that edge.
  The result is a tree with at least $k$ leaves, since we included at least $k$ leaves out of $v$.

  We now assume the graph is connected and each node has degree at most $k-1$.
  We are going to select a subgraph of $G^{*}$ with $O(k^4)$ nodes on which we will find a tree with at least~$k$ leaves using the fastest static algorithm for the problem.
  If $G^{*}$ has at most $4k^2 + 12k + 8$ vertices, we simply use $G^{*}$ (in this case, if we find that $G^{*}$ does not have a tree with~$k$ leaves, then we return that there is no such tree). 
  Otherwise, pick an arbitrary node $r$, and perform a breadth-first search from $r$ until we have reached at least $4k^2 + 12k + 8$ vertices.
  Let $S'$ be the set of those vertices, and let $S$ be the set containing $S'$ and also all the vertices at distance at most 2 from $S'$.
  Since each vertex in~$G^{*}$ has degree less than $k$, set~$S$ has $O(k^4)$ vertices.
  Let~$G_S$ be the induced subgraph of $G^{*}$ on the vertices of~$S$.
  No vertex in $S'$ is useless in $G_S$, since they are not useless in $G^{*}$, and we included enough vertices to make sure they are not useless in $S$.
  By Proposition \ref{MLSTKernel}, $G_S$ has a spanning tree $T_S$ with at least $k$ leaves, which we can find in $O(3.72^k)$ time using the static algorithm by Daligault et al.~\cite{DaligaultEtAl2010} for {\sc Max Leaf Spanning Tree}.

  We now want to merge $T$ and $T_S$ into a spanning tree with at least $k$ leaves.
  We will include all edges in~$T_S$, and all edges in $T$ which are between pairs of vertices not in $S$.
  To add edges between vertices in~$S$, and nodes not in $S$, we can use our dynamic tree structure~$D$ maintaining $T$.
  We remove from $D$ all the edges with an endpoint in $S$, and add in all the edges in $T_S$. 
  Then, for each edge in $T$ with exactly one endpoint in $S$, we use $D$ to check whether adding that edge would create a cycle in the tree we are building, and if not, we add it.
  Since each node has degree $O(k)$, this requires $O(k^5)$ operations on $D$, for a total amortized time of $O(k^5 \log n)$.
  Since the resulting tree is an extension of $T_S$, it has at least $k$ leaves, as desired.
\end{proof}

\section{Undirected \texorpdfstring{$k$}{k}-Path}
We are also able to design dynamic algorithms for some problems which do not use either the kernelization technique or the branching tree technique.
One example is our algorithm for {\sc Undirected $k$-Path}, which is the following problem:

\defparproblem{Undirected $k$-Path}{$k$}{An undirected graph $G$ and an integer $k\in\mathbb N$.}{Does $G$ have a path of length at least $k$?}

By using a color-coding technique (as introduced by Alon et al.~\cite{AlonEtAl1995}), we are able to reduce the problem to a dynamic connectivity problem, which can be solved using known dynamic data structures.
This way, we achieve an update time of $f(k)\cdot \log^{O(1)}(n)$.

Our construction is as follows.
An \emph{$(n,k)$-perfect hash family} is a set of functions $\{h_1,\ldots,h_T\}$ such that $h_i:\{1,\hdots,n\}\ldots \{1,\hdots,k\}$ for all $i$ and for every $S\subseteq \{1,\hdots,n\}$ with $|S|=k$ there is some $h_i$ that maps all the elements of $S$ to distinct colors.

\begin{proposition}[\cite{NaorEtAl1995}]
\label{thm:perfecthash}
  For any $n,k\in\mathbb N$ one can construct an $(n,k)$-perfect hash family of size $T = e^kk^{O(\log k)}\log n$ in time $e^kk^{O(\log k)}\log n$.
\end{proposition}

We will keep $k! T$ dynamic graphs $G_{i,\pi}$ for each of the $T$ choices of $h_i$ and each of the~$k!$ permutations~$\pi$ of $\{1,\hdots,k\}$.
In $G_{i,\pi}$ we associate the vertices $V$ with $\{1,\hdots,n\}$ in the natural way and we color every vertex~$v$ with $\pi(h_i(v))$.
For every edge $\{u,v\}$ of $G$ we also add it to $G_{i,\pi}$ if $\pi(h_i(v))=\pi(h_i(u))+1$, i.e., we only leave edges between adjacent colors.
We also add two vertices $s$ and~$t$ and we add edges $\{s,v\}$ for all $v$ with $\pi(h_i(v))=1$ and $(u,t)$ for all $u$ with $\pi(h_i(u))=k$.

We claim that $G$ contains a path on $k$ vertices if and only if $s$ and $t$ are connected in one of the graphs~$G_{i,\pi}$.
To see this, first suppose that $s$ and $t$ are connected in some $G_{i,\pi}$.
Then the distance between $s$ and $t$ in~$G_{i,\pi}$ is at least $k+1$ since we have a layered graph and in order to get from $s$ to $t$ the path needs to visit at least one vertex from every color class.
Moreover, every path from~$s$ to $t$ that visits~$s$ and $t$ exactly once is of the form $s$, followed by a path $P$ in $G$ and then followed by $t$, and every path from $s$ to $t$ contains a subpath from $s$ to $t$ of this form.
As the distance between $s$ and $t$ is at least $k+1$, $|V(P)|\geq k$ and we can return any subpath of $P$ on $k$ nodes.

Now suppose that $G$ has a path $P=v_1\rightarrow v_2\rightarrow\ldots\rightarrow v_k$ on $k$ vertices.
Let $h_i$ be the hash function that maps $\{v_1,\ldots,v_k\}$ to all distinct colors; we know it exists.
Let $\pi$ be the permutation of $\{1,\hdots,k\}$ that maps~$h_i(v_j)$ to $j$ for each $j$.
Then all edges of $P$ are in $G_{i,\pi}$, $v_1$ is colored $1$ and $v_k$ is colored $k$, and~$s$ followed by~$P$ followed by $t$ connects $s$ to $t$.

Hence we maintain all graphs $G_{i,\pi}$ using the fully dynamic connectivity data structures alluded to in Proposition~\ref{thm:dynamicconnectivity}.

When we insert/delete an edge $\{x,y\}$ in our dynamic data structure for {\sc $k$-Path}, we insert/delete $\{x,y\}$ into/from each graph $G_{i,\pi}$ for which $\pi(h_i(y))=\pi(h_i(x))+1$ or $\pi(h_i(y))=\pi(h_i(x))-1$ ($\{x,y\}$ can be used in either direction).
This takes time bounded by $k!2^{O(k)}k^{O(\log k)} O((\log n \log\log n)^2)$.

\section{Dense Subgraph in Bounded-Degree Graphs}

\defparproblem{Dense Subgraph in Bounded-Degree Graphs}{$k$}{A graph $G$ with maximum degree $\Delta$ and an integer $k\in\mathbb N$.}{Find an induced subgraph on $k$ vertices in $G$ with the maximum number of edges.}

Similar to our {\sc $k$-Path} algorithm, we will use a derandomized variant of color-coding.

Let us use $(n,k')$-perfect hash families for $k' = k(\Delta+1)$.
According to Proposition~\ref{thm:perfecthash}, we have $T\leq e^{k'}{k'}^{O(\log k')}\log n$ functions $h_1,\ldots, h_T$ from $\{1,\hdots,n\}$ to $\{1,\hdots,k'\}$ such that for every $S\subset \{1,\hdots,n\}$ of size at most $k'$, one of the $h_i$ maps the elements of $S$ to distinct colors.
We will maintain $2^{k'} T \leq (2e)^{k'}{k'}^{O(\log k')} \log n$ data structures $D_{i, U}$, one for every $h_i$ and every $U\subseteq \{1,\hdots,k'\}$.

Let $G$ be a graph that is input to the {\sc Dense Subgraph in Bounded-Degree Graphs} problem.
In~$D_{i,U}$ we will consider the sets $L=\{u\in V(G)~|~h_i(u)\in U\}$ and $R=\{u\in V(G)~|~h_i(u)\notin U\}$.
Let~$H$ be a subgraph of $G$ on $k$ vertices with the maximum number of edges.
Let $N(H)$ be the set of neighbors of $H$ that are not in $H$.
Then $|H\cup N(H)|\leq k(\Delta+1)=K$ since the graph has degree bounded by $\Delta$.
Hence there is some~$h_i$ which colors all the vertices of $H\cup N(H)$ into different colors.
Let $U$ be the set of colors that $h_i$ assigns to the vertices of $H$.
Consider the data structure $D_{i,U}$: in it $H\subseteq L$ but $N(H)\subseteq R$.
We will exploit this fact.

Now consider any connected component $C$ in the subgraph of $G$ induced by $L$, given that $H\subseteq L$ but $N(H)\subseteq R$.
If $C$ contains a vertex of $H$ then it cannot contain any vertices that are not in $H$, as otherwise there would be a node of $N(H)$ in $L$.
Thus, $H$ is a union of a subset of the connected components of $L$. 

Suppose that we can maintain the connected components of $L$.
In addition, for each connected component~$C$ of $L$ we will maintain the number of edges
touching $C$ (i.e., those inside $C$ and those from~$C$ to $R$) and the sizes $|C|$.
Now we need to be able to dynamically compute a collection of connected components $\{C_1,\ldots,C_\ell\}$ such that $\sum_i |C_i|=  k$ and $\sum_{i} m(C_i)$ is maximized.

Here is how we do it.
For every $D_{i,U}$ we create a dynamic connectivity data structure for the subgraph induced by $L$ using the construction of Proposition~\ref{thm:dynamicconnectivity}.
This data structure performs edge updates in time $O(\log n (\log\log n)^2)$, and given any node $u$ can return the name $C$ of the component containing~$u$.
It is straightforward to augment the data structure to also keep track of $m(C)$ for each connected component~$C$. 

In addition to this data structure, we keep a matrix $A$ of size $k \times k\Delta$ such that $A[i][j]$ contains a list of the names of components $C$ in $L$ for which 
$|C|=i$ and $m(C)=j$.
Given this matrix, we can determine the densest $k$-subgraph in $G$ as follows.
Let $p: \mathbb{N} \to \mathbb{N}$ denote the partition number, so that $p(n)$ is the number of integer partitions of $n$.
For every choice of ways to split $k$ into a sum of $\ell \leq k$ positive integers (there are~$p(k)$ such choices) $k_1,\ldots,k_\ell$, for every integer $M$ between $0$ and $k\Delta$, and every choice of a way to split~$M$ into an ordered $\ell$-tuple of non-negative integers $m_1,\ldots,m_\ell$ (such that $\sum_i m_i=M$ and $m_i\geq 0$; (there are at most $p(k \Delta)$ such choices), check whether $A[k_i][m_i]$ contains a non-empty list for every $i\leq \ell$.
This finds the largest $M$ for which there is a subgraph in $G$ on $k$ nodes and at least $M$ edges, and one can return such a subgraph by picking the first component from each of the non-empty lists $A[k_i][m_i]$.
The query time is thus $O(p(k) + k\Delta + p(k\Delta))$.
As $p(n) = 2^{O(\sqrt{n})}$ \cite{HardyRamanujan1918}, we can bound this query time by $2^{O(\sqrt{k \Delta})}$.

Let us see how the updates work.
Suppose that an edge $\{x,y\}$ is inserted (deleted).
We iterate through all choices for $D_{i,U}$.
For each $D_{i,U}$, if $x\in L$ and $y\in R$, then query the connectivity data structure for~$L$ to find the connected component $C$ that contains $x$, and increment (decrement in case of deletions)~$m(C)$.
We assume that each component has a pointer to its position in one of the lists in $A$, say $A[i,j]$.
Now move the component to the list $A[i, j+1]$ (or $A[i,j-1]$ in case of deletions), updating the pointer.
(It is similar if $x\in R$ and $y\in L$.)

If $x,y\in R$, then we do nothing.
Finally, if both $x,y\in L$ and $(x,y)$ is inserted, then we first find the components $C_1$ and $C_2$ such that $x\in C_1$, $y\in C_2$.
If these components are the same, we just insert $(x,y)$ in the connectivity data structure and move $C_1$ from $A[i][j]$ to $A[i][j+1]$ following pointers as before (the connectivity structure will maintain $m(C_1)$ automatically).
If $C_1\neq C_2$, remove both $C_1$ and $C_2$ from $A$, insert $(x,y)$ into the connectivity data structure, obtain the new component $C$ that contains both $x$ and $y$ and its $m(C)$ and then insert $C$ into $A[|C|,m(C)]$.
Note that if we ever get $|C|>k$, we do not need to insert into $A$ since the component becomes irrelevant.
If $x,y\in L$ and $(x,y)$ is deleted, obtain $C$ that $x$ and $y$ are both in, remove $C$ from $A$, delete $(x,y)$ from the connectivity structure, obtain the new components and their $m(\cdot)$ values and update the array $A$ as necessary.
The update time is $O(\log n(\log\log n)^2)$ per $D_{i,U}$.

The update time is $2^{O(\Delta k)}\log n (\log\log n)^2$ and the query time is $2^{O(\sqrt{n})}+2^{O(\Delta k)}\log n = 2^{O(\Delta k)}\log n$.

\section{Edge Clique Cover}
\label{sec:edgecliquecover}

For a graph $G$, an \emph{edge clique cover} is a set $\mathcal{C} = \{C_1,\hdots,C_k\}\subseteq V(G)$ of subgraphs of $G$, such that each~$C_i$ induces a clique in $G$ and every edge of $G$ belongs to some subgraph $C_i$.

\defparproblem{Edge Clique Cover}{$k$}{A graph $G$ and an integer $k\in\mathbb N$.}{Find an edge clique cover of $G$ with size at most $k$.}

\begin{theorem}
  There is an algorithm which can handle edge insertions or edge deletions in $2^{2^{O(g(k)}}$ time, under the promise that the graph always has an edge clique cover of size $g(k)$ for some computable function~$g$, which maintains an edge clique cover in the sense that, any of the following queries can be supported
  \begin{itemize}
    \item Given a clique name $C$, return the set of nodes in $C$ in linear time in the size of $C$.
    \item Given a node $v$, return the names of all cliques that $v$ is contained in, in linear time in the size of the output.
  \end{itemize}
\end{theorem}
\begin{proof}
  We first describe the static rules that we use; they lead to a kernel with at most $2^k$ nodes, and are due to Gy\'{a}rf\'{a}s~\cite{Gyarfas1990}.
  Namely, we repeatedly apply the following reduction rules which do not change the size of an minimum clique cover of the graph:
  \begin{itemize}
    \item Remove any isolated node.
    \item For any two nodes $u,v$ with same closed neighborhood $N[u] = N[v]$, remove $v$ from the graph (`merging' $u$ and $v$)
  \end{itemize}
  The first reduction rule is fine as isolated nodes do not belong to any cliques.
  For the second rule, whatever set of cliques we eventually put $u$ in, we can also put $v$ in those same cliques, and this will be valid since~$u$ and $v$ had the same closed neighborhoods.
  Exhaustively applying these rules to a graph~$G$ until no more are possible yields a {\em reduced graph} $G^{*}$.
  In an edge clique cover of a reduced graph, no two nodes with any edges can be in the exact same set of cliques, as their closed neighborhoods differ.
  As $G^{*}$ has no isolated nodes, if it has an edge clique over of size $s$ then it has at most $2^s$ nodes, one for each subset of the cliques that it may be in.
  Hence, our promise implies that $G^{*}$ will always have size at most~$2^{g(k)}$ throughout the sequence of edge updates.

  We will maintain $G^{*}$ along with, for each node $u$ in $G^{*}$, a list $L(u)$ of the nodes in $G$ which have been merged into $u$, stored as a doubly-linked list, and for each node 
$v$ in $G$, a pointer $p(v)$ to the location of $v$ in the list that it is in. 
  After each update, to answer a query, we will compute an edge clique cover of $G^{*}$ in $2^{2^{O(k)}} + O(2^{4g(k)})$ time using the static algorithm by Gramm et al.~\cite{GrammEtAl2009}, as $G^{*}$ has at most $2^{g(k)}$ nodes.
  These are sufficient to recover the edge clique cover of the entire graph (or any of the types of queries mentioned in the theorem statement) in linear time in the size of the answer to the query.

  It remains to describe how to maintain $G^{*}$, the lists $L(u)$, and the pointers $p(v)$.   
  Suppose an edge between nodes $a$ and $b$ is changed (either added or removed).
  We first remove $u$ and $v$ from $G^{*}$.
  We describe how to do this for $u$, then do the same for $v$.
  If $u$ was the only element of $L(p(u))$, then we do this by deleting~$u$ from $G^{*}$. 
  Otherwise, we remove $u$ from the list $p(u)$ that contained it.
  If the node corresponding to $p(u)$ was actually $u$, then we rename it to one of the other nodes remaining in $p(u)$.
  We then do the same for~$v$.
  The resulting graph is now the reduced graph on $G \setminus \{ u,v \}$, the graph $G$ when $u$ and $v$ are removed.
  Moreover,~$u$ is adjacent to a node $x$ from $G$ if and only if it is adjacent to every node in $p(x)$, since whether or not $u$ is adjacent to any node still in our $G^{*}$ has not changed, and similarly for~$v$.

  We now need to add $u$ and $v$ back in.
  Again, we describe how to do this for $u$.
  For each of the at most $2^{g(k)}$ nodes in the current graph $G^{*}$, we check whether $u$ is adjacent to it.
  If there is any node $x$ in~$G^{*}$ which is adjacent to $u$ and has the exact same neighborhood (in~$G^*$) as $u$, we add in $u$ to $L(x)$ and set $p(u) = L(x)$.
  Otherwise, we create a new node $u$ in~$G^*$, set $p(u) = L(u) = u$, and make that node adjacent to each node that $u$ is adjacent to.
  Once we do this for $v$ as well, our graph $G^{*}$ will be a valid reduced graph for the new $G$.
  The entire process took only $f(k)$ time, since $G^{*}$ has at most $2^{g(k)}$ nodes, and we only did graph operations on~$G^{*}$.
\end{proof}

\section{Undirected Feedback Vertex Set}
\defparproblem{Undirected Feedback Vertex Set}{$k$}{A graph $G$ and an integer $k\in\mathbb N$.}{Find a set $S\subseteq V(G)$ of size at most $k$ for which $G - S$ is a forest?}

We build on the degree-based algorithm described by Cygan et al.~\cite[p. 57]{CyganEtAl2015}.
We first show that any feedback vertex set $S$ of size at most $k$ must intersect the first $f(k)$ vertices of largest degree in any connected graph of minimum degree at least 2, by proving a modification of a Lemma by Iwata \cite{Iwata2017}.

\begin{lemma} 
\label{bounds}
  Let $G$ be a graph of minimum degree at least 2 without self-loops, but multi-edges are allowed.
  Let $S \subseteq V(G)$ be a feedback vertex set of $G$ of size at most $k$, and let $D$ be an upper bound on the average degree of the vertices in $S$.
  If $G$ has at most $|V(G)|/2-k$ vertices of degree 2, then $|V(G)| \leq 2Dk$ and $|E(G)| \leq 3Dk$.
  In particular, the average degree of $G$ is at most $3k$.
\end{lemma}
\begin{proof}
  As $V(G) \setminus S$ is a forest, the number of edges inside $V(G)\setminus S$ is at most $|V(G)|-k-1$.
  Hence, since at most $|V(G)|/2-k$ vertices have degree 2 and all other vertices have degree at least 3, the number of edges between $S$ and $V(G) \setminus S$ is at least
  \begin{equation*}
    3(|V(G)|-k) - (|V(G)|/2-k) - 2(|V(G)|-k-1) \geq |V(G)|/2 \enspace .
  \end{equation*}
  However, since the vertices in $S$ have average degree at most $D$, there are at most $Dk$ edges from $S$ to~$V(G) \setminus S$.
  Combining the two bounds gives that $Dk \geq |V(G)|/2$, which rearranges to $|V(G)| \leq 2Dk$.

  The sum of degrees of vertices in $G$ is $2|E(G)|$, and the sum of degrees of vertices in $S$ is at most~$Dk$, so the sum of degrees of vertices in $V(G) \setminus S$ is at least $2|E(G)| - Dk$.
  Hence, the number of edges between $S$ and $V(G) \setminus S$ is at least $2|E(G)| - Dk - 2(|V(G)|-k-1)$.
  As before, this amounts to at most $Dk$, so we get that
  \begin{multline*}
    Dk \geq 2|E(G)| - Dk - 2(|V(G)|-k-1)
       \geq 2|E(G)| - Dk - 2(2Dk-k-1)\\
          = 2|E(G)| - 5Dk +2k +2
       \geq 2|E(G)| - 5Dk  \enspace .
  \end{multline*}
  This rearranges to $|E(G)| \leq 3Dk$.

  The average degree of $G$ is at most $3k$ since $D \leq |V(G)|$, so $|E(G)| \leq 3Dk \leq 3|V(G)|k$.
\end{proof}

We are now ready to prove our main result of this section.
See the description in Sect.~\ref{subsec:examples} for a high-level overview of the algorithm.

We give a dynamic branching tree algorithm for the problem.
To make the proof easier to read, we present an inductive argument, where we define an algorithm $\mathcal A_k$ which works on instances with parameter~$k$, and which uses $\mathcal A_{k-1}$ as a subroutine.
This corresponds to a branching tree algorithm where~$\mathcal A_k$ is used at levels of the branching tree corresponding to parameter $k$.

\begin{theorem}
  There is a dynamic fixed-parameter algorithm for {\sc Feedback Vertex Set} that handles edge insertions and edge deletions in $f(k) \cdot \log^{O(1)}(n)$ amortized time for $n$-vertex graphs.
\end{theorem}
\begin{proof}
  We prove that for each $k \geq 0$ there is a dynamic algorithm $\mathcal A_k$ which handles edge insertions and edge deletions in $f(k) \cdot \log^{O(1)}(n)$ amortized time for $n$-node graphs.
  In particular, starting with an edge-less graph on $n$ nodes, $\mathcal A_k$ handles updates such that after $u$ updates, the total runtime is $O(u \cdot f(k) \cdot \log^{O(1)}(n))$.
  We prove this by induction on $k$.

  For $k=0$, this is the problem of dynamically maintaining whether a graph has a cycle.
  We solve this using the dynamic algorithm from Proposition~\ref{thm:dynamicconnectivity} for maintaining the connected components of an undirected graph, together with a count of the number of edges within each component.

  We now want to design~$\mathcal A_k$ for $k \geq 1$ given $\mathcal A_{k-1}$ as prescribed.
	
	\subsection{Maintaining the Kernel}
	
  Starting with a graph~$G$, we begin by preprocessing it.
  First, to each vertex $v$ of $G$ we assign a tree $T_v$, and to each edge~$e$ of $G$ we assign a tree~$L_e$.
  Initially, $T_v$ consists of the single vertex $v$, and when $e = \{x,y\}$ then $L_e$ consists of the vertices~$x$ and~$y$ connected by an edge.
  The set $\{T_v\}_{v \in V(G)}\cup\{L_e\}_{e \in E(G)}$ will be stored in a dynamic tree data structure $D$ that we describe in Section~\ref{sec:dynamictree}.
  We will always store $O(n+m)$ vertices in $D$, and so operations will take amortized time $O(\log n)$. 
  Note in particular that if a vertex $x$ has degree $d$, then there are $d+1$ copies of $x$ in $D$. 
  For notation, we refer to the copy of $x$ in $T_x$ as $x_x$, and the copy of $x$ in $L_{x,y}$ as $x_y$.

  For a vertex $v$, $T_v$ initially contains just $v$, and it will correspond to an induced subgraph of $G$ which is a tree rooted at $v$ that was contracted into $v$.
  For an edge $\{x,y\}$, $L_{x,y}$ initially consists of the vertices~$x$ and $y$ connected by an edge, and will eventually correspond to an induced subgraph of $G$ which is a tree containing both $x$ and $y$ that was contracted into~$\{x,y\}$.
  Importantly, $x$ and $y$ will both always be leaves in the tree corresponding to~$L_{x,y}$.

  For two adjacent vertices $x,y$ in $G$, we define the \texttt{combine(x,y)} subroutine as follows:   
  Let $a$ be the unique neighbor of $x_y$ in $L_{x,y}$.
  First cut $a$ from $x_y$ in $L_{x,y}$, then link $a$ to $x_x$ in~$T_x$. 

  We repeatedly perform the following steps until no more are possible:
  \begin{itemize}
    \item Delete any vertex $v$ of degree one.
      When doing this, let $x$ be the unique neighbor of $v$ in~$G$.
      We then need to `merge' $T_x$, $L_{x,v}$, and $T_v$ all into $T_x$.
      We do this by merging $v_v$ with~$v_x$, so that $v_x$ now has all the neighbors of both nodes, and then similarly merging $x_x$ with $x_v$, and calling then entire new connected tree $T_x$.
    \item Delete any vertex $v$ of degree two, replacing it with an edge between its two neighbors~$x$ and $y$.
      Similar to above, we want to merge $L_{x,v}$, $T_v$, and $L_{v,y}$ all into the new $L_{x,y}$.
      We do this by merging $v_v$ with both $v_x$ and $v_y$ as above, and calling the resulting tree $L_{x,y}$. 
  \end{itemize}

For any particular contracted vertex, we have a pointer into its position in the corresponding tree that it was contracted into (each contracted vertex is in exactly one tree of~$D$).
This can be done in $O(m\log n)$ time using $D$ and by maintaining the vertices' degrees throughout the process.
In particular this takes no time on an empty initial graph.
The resulting graph~$G^*$ may have multi-edges or self-loops.
Note in particular that the size of a minimum FVS has not changed, since whenever we deleted a vertex, we did not need to include it in a minimum FVS, and moreover, any minimum FVS for $G^*$ works for $G$ as well.
Let $m^*$ be the number of edges in $G^*$, and note that $m^* \leq m$.

We first describe how we will reflect updates to $G$ in our representation $G^*$.
We will maintain that the set of vertices for a minimum FVS of $G$ still appears in $G^*$, and that each update to $G$ will cause at most~9 updates to $G^*$.
Each vertex in $G$ will either also be in $G^*$, or else it will be in at least one $L_e$ or $T_v$, and in this case it won't have any other neighbors in $G$ besides the vertices it is adjacent to in that tree.
  
Suppose an edge is inserted/deleted between vertices $u$ and $v$ of $G$.
First, we reintroduce~$u$ and $v$ into~$G^*$ if they are not there.
Suppose $u$ is not in $G^*$ (and then deal with $v$ symmetrically).
Hence, $u$ appears in $D$.
There are three cases depending on where $u$ is in $D$.

  \begin{itemize}
    \item If $u$ is in $T_v$ for some vertex $v$ in $G^*$, then we add $u$ and an edge from $v$ to $u$ into $G^*$.
      Let~$w$ be the first vertex on the path from $u$ to $v$ in $T_v$.
      We cut between $u$ and $w$ in~$T_v$, then rename $u$ to $u_u$.
      We create a new vertex~$u_v$ which we link to $w$.
      Next let $x$ be the first vertex on the path from $v_v$ to $w$.
      We cut between~$v_v$ and $x$, then create a new vertex~$v_u$ which we link to $x$.
      The component which still contains~$v_v$ is still $T_v$, the component containing~$u_u$ is now $T_u$, and the remaining component containing $v_u$ and~$u_v$ is now $L_{u,v}$.
    \item If $u$ is in $L_{x,y}$ for some vertices $x,y$ in $G^*$, then we have two cases depending on whether~$u$ is on the path from $x$ to $y$ in $L_{x,y}$.
      To determine this, we root $L_{x,y}$ at $u$, and find the nearest common ancestor~$z$ of~$x$ and $y$.
      \begin{itemize}
        \item If $z = u$, then $u$ is on the path from $x$ to $y$.
          In $G^*$, we delete the edge between $x$ and~$y$, and replace it with a new vertex $u$, and edges between $x$ and $u$, and between $y$ and~$u$.
          Let $b$ be the first vertex on the path from $u$ to $x$ in $L_{x,y}$.
          Cut $(b,u)$, create a new vertex~$u_x$, and link $u_x$ to $b$.
          Then, rename $x_y$ to $x_u$.
          The new component containing $x_u$ is now $L_{u,x}$.
          Do the same with $y$ to create the new $L_{u,y}$.
          The remaining component in $D$ containing the original $u$ from $L_{x,y}$ is now renamed to $T_u$.
        \item If $z \neq u$ then $u$ is not on the path from $x$ to $y$.
          First we do the previous steps with~$u$ replaced by~$z$.
          We have thus `split' edge $(x,y)$ by putting vertex $z$ in the middle.
          Furthermore, $u$ is now in~$T_z$.
          We then do the first step with $v = z$ to add a vertex $u$ adjacent to only $z$.
      \end{itemize}
  \end{itemize}

  We note that these rules work even if $u$ appears in $L_e$ where $e$ is a self-loop on a vertex~$w$.
  In this case we apply the rules where $x$ and $y$ are the two copies of $w$ in $L_e$.

  The two vertices $u$ and $v$ now appear in $G^*$, and we can make the new edge update as appropriate.
  Finally, some vertices may have degree two or one and need to be removed, which can cause more vertices to need to be removed, and so on.
  We do this in the way we originally did to construct $G^*$.
  In particular, throughout this process we added at most four vertices to $G^*$, namely $u$, $v$, and the two additional vertices we may have added to the graph ($z$ in the case $z \neq u$), and these are the only vertices which could have degree one.
  Deleting a vertex of degree one is the only way to decrease the degree of any other vertex.   
  Moreover, every vertex other than these four originally had degree at least three.
  Hence, this process will remove at most eight vertices (these four, which can reduce two more vertices to degree one, which can in turn decrease one more vertex to degree one, and in turn one more vertex to degree two).
  In all, at most 17 edges were changed (1 given edge change, 8 edge changes when simplifying $G^*$, and 8 edge changes when adding $u$ and $v$ to $G^*$).

	\subsection{Choosing vertices to branch on}
  Let us now describe what $\mathcal A_k$ will do assuming $G^*$ has no self-loops.
  We will then describe the modification if $G^*$ does have self-loops.

  Set $\delta = m^* / (6k+1)$, and let $B_H$ be the set of vertices in $G^*$ which have degree at least $\delta$ (initially; this set will not change as we make updates to $G^*$).
  By Lemma \ref{bounds}, we know that if $G^*$ has a FVS of size at most~$k$, then it must include a vertex of degree at least $m^* / (3k)$.
  Hence, one of the vertices in $B_H$ need to be included in such an FVS.
  Moreover, even if we make up to $\delta$ edge insertions and deletions, yielding a graph with $m' \geq m^* - \delta = m^* \cdot \frac{6k}{6k+1}$ edges, the set of vertices of degree at least $m'/(3k) \geq 2 \delta$ must be entirely contained within~$B_H$, since vertices not in $B_H$ currently have degree at most $\delta - 1$, and so after $\delta$ insertions they can have degree at most $2\delta - 1$.
  Moreover, Lemma \ref{bounds} still applies, since each edge update can create at most two new vertices of degree two, so the number of degree-2 vertices is at most $2\delta < n/2 - k$.
  In other words, as long as we make at most $\delta$ edge modifications to~$G^*$, and never introduce self-loops, we know that if the current $G^*$ has a FVS of size at most~$k$, then it must intersect~$B_H$ in at least one vertex.

  We will branch on picking each vertex of $B_H$ to be in our FVS.
  Since $|B_H| \leq 2m^* / \delta = 12k+1$, we are branching on at most $12k+1$ choices.
  For each $v \in B_H$, we initialize $\mathcal A_{k-1}$ on the graph $G^* \setminus \{ v \}$.
  Whenever there is an update to $G^*$, we will reflect that update in the $G^* \setminus \{ v \}$ for $A_{k-1}$ as well.
  Then, if~$A_{k-1}$ says there is a FVS $S$ of size at most $k-1$, we can return $S \cup \{ v \}$ as our FVS of size $k$.
  If all of the branches say that there is no FVS of size $k-1$, then we know that there is no FVS in $G$ of size $k$ either.
  This is guaranteed to be correct as long as there have been at most $\delta$ updates to $G^*$.

  After $\delta$ updates, we simply reinitialize on the current graph $G$, by recomputing $G^*$ and~$B_H$, and the recursive calls.
  This way we are guaranteed that the algorithm is always correct.
  This time cost will be amortized over the $\delta$ updates, as we will discuss later.

  Suppose instead that $G^*$ initially has self-loops.
  Recall that our process for updating $G^*$ upon edge insertions or deletions can never create new self-loops in $G^*$, and so we can only have self-loops if they are initially present in $G^*$.
  Let $B_L$ be the set of nodes with self-loops (this set is computed initially in $O(m)$ time).
  Any FVS for $G^*$ must contain all of $B_L$.
  If $|B_L| \leq k$, then it is still possible that there is a FVS of size at most $k$.
  In this case, instead of branching on only $B_H$, we instead branch on $B := B_H \cup B_L$.
  As long as we have done at most $\delta$ edge updates, then if there are still any self-loops in our graph, any FVS must contain a node (in fact, every vertex) in $B_L$, and if the updates have removed all the self-loops, then a FVS of size at most $k$ must contain a vertex from $B_H$.
  Hence, as long as we have done at most $\delta$ edge updates, a FVS of size at most $k$ must contain a vertex from $B$.
  We hence proceed as before, branching on all $|B| \leq 13k+1$ vertices.

  If, on the other hand, $|B_L| > k$, then we know that there is no FVS of size at most $k$, and we return this.
  Whenever an edge update to $G^*$ deletes a self-loop, and that node no longer has a self-loop, we remove it from $B_L$.
  Once we get to $|B_L| \leq k$, it becomes possible that there is a FVS of size at most $k$, so we then compute $B$ and continue as in the above paragraph instead.

  This concludes the description of the algorithm and proof of correctness.
  We now analyze the runtime.
  While $|B_L| > k$, each update takes only constant time to update at most 17 edges and hence remove at most 17 nodes from $B_L$.
  The initialization, either at the beginning or once $|B_L| \leq k$, takes $O(m \cdot (1 + (13k+1) \cdot f(k-1)) \log^{O(1)}n)$ time, since computing $G^*$ and $B$ takes $O(m\log n)$ time, and then we initialize $13k+1$ copies of $\mathcal A_{k-1}$.
  Each update before $\delta$ updates to $G^*$ causes 17 updates to $G^*$, and then we need to propagate those updates to the $13k+1$ copies of $\mathcal A_{k-1}$, which takes $O((1 + (13k+1) \cdot f(k-1)) \log^{O(1)}n)$ time.
  Reinitialization after $\delta$ updates to $G^*$ again takes $O(m \cdot (1 + (13k+1) \cdot f(k-1)) \log^{O(1)}n)$ time, but amortized over all $\delta/17 = m^* / (17 (6k+1)) \geq m (1 - 1/(6k+1)) / 17(6k+1))$ updates, takes only amortized time $O((1 + \frac{(6k+1)(13k+1)}{1 - 1/(6k+1)}) \cdot f(k-1)) \log^{O(1)}n)$ per update.
  These are of the desired form with $f(k) = O((17 \cdot 6)^k \cdot (k!)^2) = 2^{O(k \log k)}$.
\end{proof}

\section{Conditional Lower Bounds}
Our conditional lower bounds are based on two hypotheses involving reachability oracles (ROs): data structures that preprocess a given directed graph $G$ and are able, for any queried pair of vertices $u,v$, to return whether $u$ can reach~$v$ in $G$.
Suppose that $G$ has $m$ edges and $n$ vertices.
There are essentially only two approaches to this problem: (1) precompute and store the transitive closure of $G$ in $O(\min\{n^\omega,mn\})$ time (for $\omega<2.373$~\cite{VassilevskaWilliams2012}) and then answer queries in constant time; (2) do not precompute anything and answer queries in~$O(m)$ time.
If the current algorithms are actually optimal, then this would mean that any reachability oracle with $O(n^{2-\eps})$ preprocessing time for $\eps>0$ must require $m^{1-o(1)}$ time to answer queries.
We are not going to make such a bold conjecture, even though it is currently plausible.
Instead we will use the much more innocent hypothesis that an RO that uses $O(m^{1+o(1)})$ preprocessing time, cannot answer reachability queries in $m^{o(1)}$ time.

\begin{hypothesis}
  On a word-RAM with $O(\log m)$ bit words, any Reachability Oracle for directed acyclic graphs on $m$ edges must either use $m^{1+\eps}$ preprocessing time for some $\eps>0$, or must use $\Omega(m^\delta)$ time to answer reachability queries for some constant $\delta>0$.
\end{hypothesis}

P{\v{a}}tra{\c{s}}cu studied this hypothesis, and while he was not able to prove it, he showed the following strong cell probe lower bound~\cite{Patrascu2011}: If an RO uses $n^{1+o(1)}$ words of space on a word-RAM with $O(\log n)$ bit words, then in the cell probe model it must require query time~$\omega(1)$.
His lower bound also holds for butterfly graphs (which are DAGs) with $n$ vertices and $n^{1+o(1)}$ edges, and so Patrascu's unconditional bound shows in particular:

\begin{corollary}[follows from Patrascu's work]
  There are directed acyclic graphs on $m$ edges for which any RO that uses $m^{1+o(1)}$ preprocessing time (and hence space) in the word-RAM with $O(\log n)$ bit words, must have $\omega(1)$ query time.
\end{corollary}

We will be use this unconditional result to obtain similar unconditional lower bounds for some of our dynamic parameterized problems.
Let us first see that two popular conjectures used in prior work strongly support our RO hypothesis above.

The {\sc Triangle Dection} problem asks to determine the existence of a cycle of length 3 in a given graph.
\begin{triangleconjecture}
  There is a constant $\eps>0$, so that on a word-RAM with $O(\log m)$ bit words, detecting a triangle in an $m$-edge graph requires $\Omega(m^{1+\eps})$ time.
\end{triangleconjecture}

The {\sc 3SUM} problem asks if a given set of $n$ integers contains three elements that sum to zero.
The {\sc 3SUM} problem can easily be solved in $O(n^2)$ time, and this is conjectured to be optimal:
\begin{3sumconjecture}
  There is a constant $c$, so that on a word-RAM with $O(\log n)$ bit words, $3$SUM on~$n$ integers from $\{-n^c,\ldots,n^c\}$ requires $n^{2-o(1)}$ time.
\end{3sumconjecture}

These conjectures have been used a lot (see e.g.~\cite{Patrascu2010,KopelowitzEtAl2016,AbboudVassilevskaWilliams2014}).
The Triangle Conjecture is particularly plausible, since the currently fastest algorithm for {\sc Triangle Detection} runs in $O(\min\{n^\omega,m^{2\omega/(\omega+1)}\})$ time~ \cite{AlonEtAl1997,ItaiRodeh1978}; this is quite far from linear time.
The fastest algorithm for {\sc 3SUM} on $n$ integers runs in $O(n^2 / \max\{\frac{w}{\log^2w},\frac{\log^2n}{(\log\log n)^2}\})$ time in the word-RAM model with $w$-bit words~\cite{BaranEtAl2008}.

Using a slight modification of the reductions by Abboud and Vassilevska Williams~\cite{AbboudVassilevskaWilliams2014}, we can reduce both the {\sc Triangle Detection} problem and the {\sc 3SUM} problem not only to the RO problem, but in fact to the following potentially simpler special case:

\medskip
\noindent
{\bf $k$-Layered Reachability Oracle ($k$-LRO)}:
Let $G$ be a given graph the vertices of which consist of $k$ layers of up to $n$ vertices each, $L_1,\ldots,L_k$.
The edges of $G$ are directed and go between adjacent layers in increasing order of $i$, i.e. the edge set is $\cup_{i=1}^{k-1} E_i$ where $E_i\subseteq L_i\times L_{i+1}$. 
Given $G$, one is to preprocess it so that on query $\{u,v\}$ for $u\in L_1$, $v\in L_k$, one can return whether $u$ can reach $v$ in $G$.

\medskip
\noindent
{\bf Triangle Conjecture implies hardness for $3$-LRO.}
The {\sc Triangle Detection} problem in $m$-edge graphs can be reduced to $3$-LRO on graphs with $O(m)$ edges where $m$ queries are asked, as follows.
Given a graph $G$, create a new graph $G'$ by taking three copies $L_1,L_2,L_3$ of the vertex set $V(G)$, where each copy forms a layer of 3-LRO.; the copies of $v\in V(G)$ are~$v_i$ in $L_i$ for $i=1,2,3$.
Every edge $\{u,v\}\in E(G)$ appears as four (directed) edges: $(u_i,v_{i+1})$ and $(u_{i+1},v_i)$ for $i=1,2$.
Now, for any $u,v\in V(G)$, $u_1$ can reach $v_3$ if and only if there is a path of length $2$ between $u$ and $v$ in $G$, and hence an edge $(u,v)$ of $G$ is in a triangle if and only if $u_1$ can reach $v_3$ in $G'$.
Thus, the Triangle Conjecture implies that if $3$-LRO on $m$-edge graphs can be solved with $m^{1+o(1)}$ preprocessing time, then its queries must take~$\Omega(m^\eps)$ time for some $\eps>0$.

\medskip
\noindent
{\bf The {\sc 3SUM} Conjecture implies hardness for $3$-LRO.}
The {\sc 3SUM} problem on $n$ integers can be reduced to $3$-LRO on an $O(n)$-vertex $O(n^{1.5})$-edge graph, where $O(n^{1.5})$ queries are asked.
The reduction uses a construction by Patrascu~\cite{Patrascu2010} who showed that {\sc 3SUM} can be reduced to triangle listing in a special graph.
Abboud and Vassilevska Williams~\cite{AbboudVassilevskaWilliams2014} further noticed that Patrascu's reduction reduces {\sc 3SUM} on~$n$ numbers to the following problem.
Let $\eps>0$ be any constant.
Let $G$ be a directed graph with $3$ layers $L_1,L_2,L_3$ with some edges between vertices of $L_i$ and vertices of $L_{i+1}$ for $i=1,2$.
Here, $|L_1|\leq O(n^{1+\eps}),\linebreak |L_2|\leq n, |L_3|\leq O(n^{1+\eps})$, and every vertex in $L_1$ has $\leq O(n^{1/2-\eps})$ neighbors in $L_2$; the same holds for $L_3$.
Now, we are given $n^{1.5+\eps}$ pairs $(a^j_1,b^j_3)\in L_1\times L_3$ and need to determine for each such pair whether there is some $j$ so that $a^j_1$ can reach $a^j_3$ in $G$. 

This task is just a $3$-LRO problem with $m\leq O(n^{1.5})$ edges and $O(n^{1.5+\eps})$ queries.
Therefore, if this problem can be solved in $O(m^{4/3-\delta})$ time for some $\delta>0$, then {\sc 3SUM} can be solved in $O(n^{2-1.5\delta})$ time.
Thus, assuming the {\sc 3SUM} conjecture, each of the $O(n^{1.5+\eps})$ queries needs to take $n^{0.5-\eps-o(1)}$ time (amortized), and this expression equals $m^{1/3-3\eps-o(1)}$.
Since $\eps>0$ can be chosen arbitrarily, we obtain that either the preprocessing time is at least $m^{4/3-o(1)}$, or the query time is at least $m^{1/3-o(1)}$.

We note that we can assume an even weaker form of the $3$SUM conjecture and still get a meaningful result here: suppose that $3$SUM is believed to require $\Omega(n^{1.5+\eps})$ time for some $\eps>0$, then $3$-LRO must either need $\Omega(m^{1+\delta})$ preprocessing time or $\Omega(m^\delta)$ query time (for $\delta=2\eps/3>0$).
In any case, if one believes either the triangle conjecture or the {\sc 3SUM} conjecture, then the following strengthening of our RO hypothesis should be a very popular conjecture:

\begin{hypothesis}[LRO]
  There is some constant $\eps>0$ and some constant $\ell$ so that any $\ell$-LRO for $m$-edge graphs either has $\Omega(m^{1+\eps})$ preprocessing time, or $\Omega(m^\eps)$ query time.
\end{hypothesis}
In fact, one should believe the above conjecture even for the special case of $3$-LROs.

We now show how this conjecture implies strong hardness for several parameterized dynamic problems.
\begin{theorem}
  Fix the word-RAM model of computation with $w$-bit words for $w = O(\log n)$ for inputs of size~$n$.
  Assuming the LRO Hypothesis, there is some $\delta>0$ for which the following dynamic parameterized graph problems on $m$-edge graphs require either $\Omega(m^{1+\delta})$ preprocessing or $\Omega(m^\delta)$ update or query time:
  \begin{itemize}
   \item {\sc Directed $k$-Path} under edge insertions and deletions,
   \item {\sc Steiner Tree} under terminal activation and deactivation, and
   \item {\sc Vertex Cover Above LP} under edge insertions and deletions. 
  \end{itemize}
  Under the RO Hypothesis (and hence also under the LRO Hypothesis), there is a $\delta>0$ so that {\sc Directed Feedback Vertex Set} under edge insertions and deletions requires $\Omega(m^\delta)$ update time or query time. 

  Unconditionally, there is no function $f$ for which a dynamic data structure for {\sc Directed Feedback Vertex Set} performs updates and answers queries in $O(f(k))$ time. 
\end{theorem}

\subsection{Directed $k$-Path}

\defparproblem{Directed $k$-Path}{$k$}{A directed graph $G$ and an integer $k\in\mathbb N_0$.}{Find a directed path in $G$ on $k$ vertices.}

The updates to be supported are edge insertions and deletions.
We show that {\sc Directed $k$-Path} probably does not admit fast dynamic algorithms.
Let $\ell$ be the value for which the LRO conjecture holds.
That is, under the {\sc 3SUM} Conjecture or the Triangle Conjecture, we can let $\ell=3$.
Suppose, for sake of contradiction, that there was a fully dynamic data structure $D$ for {\sc Directed $k$-Path} that can support updates and queries in $f(k) n^{o(1)}$ time; we will apply it for $k=\ell+2$.

We start from an instance $G$ of the $\ell$-LRO problem and add two extra vertices $s$ and $t$.
We insert all the edges of $G$ into $D$ in total $O(f(k)mn^{o(1)})\leq m^{1+o(1)}$ time (as $k=\ell+2$ is a constant).
When a query $(u,v)$ (for $u\in L_1,v\in L_\ell$) is given to the $\ell$-LRO data structure, we simulate it with $D$ as follows: first, add a directed edge from $s$ to $u$ and a directed edge from $v$ to $t$ in total $O(f(k)m^{o(1)})\leq m^{o(1)}$ time.
Then, ask the query whether the graph has a path on $k=\ell+2$ vertices.
Since the graph has $\ell+2$ layers $\{s\},L_1,\ldots,L_\ell,\{t\}$ and edges only go from left to right between adjacent layers, the only way for there to be a directed path on $k = \ell+2$ vertices if it uses the edges $(s,u)$ and $(v,t)$ and has a path from $u$ to $v$ in the $\ell$-LRO instance~$G$.
Thus, the query for a path on $\ell+2$ vertices answers the $\ell$-LRO query in $m^{o(1)}$ time.
After the query, we delete the edges $(s,u)$ and $(v,t)$ also in $m^{o(1)}$ time.
Thus, our algorithm can answer $\ell$-LRO queries in total~$m^{o(1)}$ time, contradicting the LRO conjecture.

\subsection{Steiner Tree}
\label{sec:steinertree}
\defparproblem{Steiner Tree}{$|T|, k$}{A connected graph $G$ and a set $T\subseteq V(G)$ of terminals; an integer $k\in\mathbb N$.}{Find a subtree in $G$ on at most $k$ edges that connects all terminals in $T$.}

The dynamic version of {\sc Steiner Tree} that we consider is where the input graph $G$ is fixed, and vertices get activated and deactivated as terminals.
That is, vertices from $V(G)$ get added to the set $\mathcal M$ and removed from it.
The query asks to return a $k$-edge tree that connects all currently active terminals.
This model has been studied by {\L}{{a}}cki et al.~\cite{LackiEtAl2015}, who provide constant-factor approximations.

We show that under the LRO conjecture, this problem does not have an efficient dynamic algorithm, even when the parameterization is under both $k$ and the number of active terminals.
Let $\ell$ be the value for which the LRO conjecture holds; that is, we may take $\ell = 3$ assuming the {\sc 3SUM} Conjecture or the Triangle Conjecture.
Suppose, for sake of contradiction, that there was a fully dynamic data structure $D$ for the special case of {\sc Steiner Tree} with 2 terminals that can support updates and queries in $f(k) n^{o(1)}$ time.
Set $k=\ell-1$.

We start from an instance $G$ of the $\ell$-LRO problem and make all its edges undirected.
We take the resulting undirected graph $G'$ as the fixed graph of the {\sc Steiner Tree} instance.
All vertices are considered inactive.
To simulate a query $(u,v)$ to the $\ell$-LRO ($u\in L_1,v\in L_\ell$) we simply activate $u$ and $v$ in $G'$.
Since the distance between $u$ and $v$ (and any pair of vertices in $L_1\times L_\ell$) is at least $\ell-1$, the only way for there to be a tree on $k=\ell-1$ edges connecting them is if this tree is a path between them that contains exactly one vertex from each $L_j$.
This path then corresponds to a {\em directed} path in $G$, answering whether $u$ can reach $v$ in $G$.
After the query is answered, $u$ and $v$ are deactivated.
Thus, under the LRO conjecture, the special case of {\sc Steiner Tree} with 2 terminals must either require $\Omega(m^{1+\eps})$ preprocessing time for some $\eps>0$, or the updates and/or queries must take $\Omega(m^\eps)$ time for some $\eps>0$.

\subsection{Directed Feedback Vertex Set}
\defparproblem{Directed Feedback Vertex Set (DFVS)}{$k$}{A directed graph $G$ and an integer $k\in\mathbb N_0$.}{Find a set $X\subseteq V(G)$ of at most $k$ vertices that hits all directed cycles of $G$.}
 
Our conditional lower bound will be based on the full Reachability Oracle hypothesis for arbitrary DAGs.
We reduce the RO problem to {\sc Directed Feedback Vertex Set} with parameter value $k = 0$; note that one can clearly make the reduction work for any constant value of $k$ by adding $k$ disjoint cycles to the reduction instance.
We are given a directed acyclic graph $G$ on $m$ edges for which we want to create a RO. 
Suppose, for sake of contradiction, that we had a dynamic algorithm for {\sc Directed Feedback Vertex Set} that performs updates and queries in $m^{o(1)}$ time.
We will create a dynamic data structure~$D_0$ for {\sc Directed Feedback Vertex Set} with $k = 0$ on the vertices of $G$, together with an extra vertex $s$.
We start by inserting all edges of $G$ into $D_0$; this takes $m^{1+o(1)}$ time.
When a reachability query $(u,v)$ comes to the RO, we simulate it by adding edges $(v,s)$ and $(s,u)$.
Now any cycle (if one exists) must use these two edges, together with a path from $u$ to~$v$ in the original DAG $G$.
There is a directed feedback vertex set of size 0 if and only if the new graph is still a DAG and hence if and only if $u$ cannot reach $v$.
The data structure~$D_0$ checks exactly whether the new graph is a DAG, and thus answers reachability queries.
After the query is answered, we remove the edges $(v,s)$ and $(s,u)$, and thus the RO queries can be answered  in $m^{o(1)}$ time, contradicting the RO hypothesis.

Because our lower bound follows from a reduction from Reachability Oracles for arbitrary DAGs, we also immediately obtain an unconditional lower bound: 
There is no computable function $f$ for which a dynamic data structure for {\sc Directed Feedback Vertex Set} performs updates and answers queries in $O(f(k))$ time. 

\subsection{Vertex Cover Above LP}
Here we consider the {\sc Vertex Cover} problem under edge insertions and deletions.
The parameter of interest is the difference $k$ between the minimum size of the vertex cover and the optimal solution of the linear programming relaxation of the natural integer programming formulation of the {\sc Vertex Cover} problem in a graph $G$:
\begin{equation*}
  \min \sum_{v\in V(G)} x_v \mbox{ subject to } x_v\in [0,1], v \in V(G) \mbox{ and } x_u+x_v\geq 1, \{u,v\}\in E(G) \enspace .
\end{equation*}

Abboud and Vassilevska Williams~\cite{AbboudVassilevskaWilliams2014} show that assuming both the {\sc 3SUM} Conjecture and the Triangle Conjecture, any dynamic algorithm that maintains a perfect matching in a bipartite graph on $m$ edges requires $\Omega(m^\eps)$ update time or query time for some constant $\eps>0$.
In other words, $m^{o(1)}$ update time and query time are very unlikely.

Consider now a bipartite graph $G$ on $n$ vertices and $m$ edges.
It is well known that in any bipartite graph, the maximum size of any matching equals the minimum size of a vertex cover.
Furthermore, if $G$ contains a perfect matching, then the optimal solution to the above natural LP is integral and equals $n/2$, the size of the perfect matching (and the minimum size of a vertex cover). 
Thus, the gap $\lambda$ between the minimum size of a vertex cover and the optimal value of the LP is $0$.
If {\sc Vertex Cover} parameterized by integrality gap~$\lambda$ has a dynamic algorithm with update time $f(\lambda)m^{o(1)}$ for any function~$f$, then bipartite perfect matching can be maintained with update/query time~$m^{o(1)}$, which contradicts the {\sc 3SUM} Conjecture or the Triangle Conjecture.

Here we give a very simple, direct reduction from the LRO problem, showing that the {\sc Vertex Cover} problem parameterized by the integrality gap $\lambda$ is hard to make dynamic even under the potentially more plausible LRO conjecture.
Let $\ell$ be the constant from the conjecture and consider an $\ell$-layered directed graph $G$ with layers $L_1,\ldots,L_\ell$.
We take every~$L_i$ and split each of its vertices $v$ into $v'$ and $v''$ with an undirected edge between them.
Take any edge $(u,v)$ incident to $v\in L_i$ such that $u\in L_{i-1}$ and replace it with the undirected edge $\{u'',v'\}$.
Add two new vertices $s$ and $t$.
Insert all these edges and vertices into the data structure for {\sc Vertex Cover} parameterized by the integrality gap $\lambda$ that supposedly handles updates in $m^{o(1)}$ time.
So far, we have spent $m^{1+o(1)}$ time.
The current $N$-vertex graph has a maximum matching size of $(N-1)/2$ composed of all edges $(v',v'')$ and where only $s$ and $t$ are unmatched.

When a reachability query $(x,y)$ for $x\in L_1,y\in L_\ell$ arrives at the $\ell$-LRO, we simulate it by adding edges $(s,x)$ and $(y,t)$.
After this, the graph has a perfect matching if and only if there is an augmenting path from~$s$ to $t$ over the original matching.
Because of the structure of the graph one can show by induction that in every layer $L_i$, at most one edge $\{v',v''\}$ can be missing from the perfect matching, so that the augmenting path must look like $s\rightarrow v'_1\rightarrow v''_1\rightarrow\ldots\rightarrow v'_j\rightarrow v''_j\rightarrow\ldots v'_\ell\rightarrow v''_\ell \rightarrow t$ where every $v_j$ is in $L_j$ for $j=1,\ldots,\ell$, $x=v_1,y=v_\ell$.
This can only happen when $x=v_1\rightarrow \ldots\rightarrow v_\ell=y$ is a directed path in $G$, and hence the reachability query is answered.
After this, the edges $(s,x),(y,t)$ are removed, and the LRO is simulated with $m^{1+o(1)}$ preprocessing time and $m^{o(1)}$ time queries, contradicting the LRO hypothesis.

\section{Dynamic Algorithms from Prior Work}
\label{sec:priorwork}

Here we expand on some dynamic algorithms from past work which we use in some of our algorithms.

\subsection{Dynamic Tree Data Structure}
\label{sec:dynamictree}

Some of our algorithms use a dynamic tree data structure and several different operations it allows. We implement this dynamic tree structure using Sleator and Tarjan's link/cut tree~\cite{SleatorTarjan1983}.
Here we highlight the operations they support which we use:

\begin{proposition}[\cite{SleatorTarjan1983}]
  There is a data structure for maintaining a forest which supports the following operations, each in $O(\log n)$ amortized time per operation with $n$ nodes in the structure:
  \begin{itemize}
    \item \texttt{maketree()} -- make a new vertex in a singleton tree.
    \item \texttt{link(a,b)} -- Adds an edge between $a$ and $b$.
    \item \texttt{cut(a,b)} -- Remove the edge between $a$ and $b$.
    \item \texttt{after(a,b)} -- Find the vertex after $a$ on the path from $a$ to $b$.
    \item \texttt{evert(a)} -- Makes $a$ the root of the tree it is in.
    \item \texttt{nca(a,b)} -- Find the nearest common ancestor of $a$ and $b$.
  \end{itemize}
\end{proposition}

\subsection{Dynamic Connectivity}
\label{sec:dynamicconnectivity}
Many of our algorithms also rely on the following dynamic algorithms for connectivity.
The expected amortized time bound is due to a recent improvement by Huang et al.~\cite{HuangEtAl2016} over an algorithm of Thorup~\cite{Thorup2000}.
The worst case expected time bounds are due to Gibb et al.~\cite{gibb2015dynamic} improving upon Kapron et al.~\cite{KapronEtAl2013}, and the deterministic amortized time bounds are due to Wulff-Nilsen~\cite{Wulff-Nilsen13a}.

\begin{proposition}[\cite{gibb2015dynamic,HuangEtAl2016,KapronEtAl2013,Thorup2000,Wulff-Nilsen13a}]
\label{thm:dynamicconnectivity}
  There is a dynamic data structure that maintains a spanning forest of a graph, answers connectivity queries in $O(\log n)$ time, and supports edge deletion and insertions in either 
	\begin{itemize}
	\item {\em expected amortized} update time $O(\log n(\log\log n)^2)$, or in
	\item {\em expected worst case}\footnote{The expected worst case data structures for dynamic connectivity from the literature assume an \emph{oblivious adversary} who does not get access to the random bits used by the data structure, so our results using $DC$ with expected worst case guarantees do as well.} time $O(\log^4 n)$, or in
	\item {\em deterministic amortized} time $O(\log^2 n/\log\log n)$.
	\end{itemize}
\end{proposition}

\section{Conclusion}
We give dynamic fixed-parameter algorithms for a variety of classic fixed-parameter problems, and lower bounds suggesting that this is not possible for others. The next step, of course, is to expand such results to more problems. Two particular problems come to mind.

First, it would be exciting to extend Bodlaender's algorithm for maintaining a tree decomposition~\cite{Bodlaender1993} to work for any treewidth parameter $k$. This would open many problems parameterized by treewidth to our approach, as most fixed-parameter algorithms for such problems need to use the tree decomposition.

Second, we are only able to design algorithms in the promise model for some problems, but depending on the specific problem, the promise model can be substantially weaker than the full model. It would be interesting to explore which problems have dynamic fixed-parameter algorithms in the promise model but not the full model.

\subparagraph*{Acknowledgments.}
The authors would like to thank Nicole Wein, Daniel Stubbs, Hubert Teo, and Ryan Williams for fruitful conversations.

\newpage

\bibliographystyle{abbrv}
\bibliography{dynamicfpt}

\end{document}